\documentclass[12pt]{article}
\usepackage{amsfonts}
\usepackage{amssymb}
\usepackage{graphicx}
\usepackage{color}
\usepackage{amsmath}
\usepackage{amsthm}
\usepackage{fullpage}
\usepackage{natbib}

\newtheorem{theorem}{Theorem}
\newtheorem{corollary}{Corollary}

\newcommand{\mP}{\mathbb{P}}
\newcommand{\E}{\mathbb{E}}
\newcommand{\gig}{\mathrm{GIG}}
\newcommand{\ig}{\mathrm{IN}}

\begin{document}

\title{\vspace{-1cm} Simulation-based Regularized Logistic Regression}

\author{Robert B. Gramacy\footnote{Part of this work was done while
    RBG was at the Statistical Laboratory, University of Cambridge} 
 \;\& Nicholas G.~Polson\\Booth School of Business\\
  University of Chicago, USA}

\date{First Draft: November 2009\\
This Draft: January 2012 \vspace{-0.35cm}}

\maketitle


\begin{abstract}
  In this paper, we develop a simulation-based framework for
  regularized logistic regression, exploiting two novel results for
  scale mixtures of normals.  By carefully choosing a hierarchical
  model for the likelihood by one type of mixture, and implementing
  regularization with another, we obtain new MCMC schemes with varying
  efficiency depending on the data type (binary v.~binomial, say) and
  the desired estimator (maximum likelihood, maximum {\em a
    posteriori}, posterior mean).  Advantages of our omnibus approach
  include flexibility, computational efficiency, applicability in $p
  \gg n$ settings, uncertainty estimates, variable selection, and
  assessing the optimal degree of regularization.  We compare our
  methodology to modern alternatives on both synthetic and real data.
  An {\sf R} package called {\tt reglogit} is available on CRAN.

\end{abstract} 

\section{Introduction}
\label{sec:intro}

Large scale logistic regression has numerous modern day applications
from text classification to genetics.  We develop a flexible framework
for maximum likelihood, maximum {\em a posteriori}, and full Bayesian
posterior inference for regularized models. Our motivations stem from
a desire to find common ground between point estimation in
``large-$p$'' settings
\citep{krish:etal:2005,genkin:lewis:madigan:2007}, where $p$ is the
number of predictors, and full Bayesian inference for ``small-$p$''
\citep{holmes:held:2006,fruh2:2007,frueh2:etal:2009,fahrmeir:etal:2010,frueh2:2010}.
Collecting such distinct methods into a unifying framework facilitates
a number of novel enhancements including posterior inference for the
amount of regularization, and an efficient handling of binomial data.

We start by framing a typical regularized optimization criteria as a
powered-up posterior, or {\em power-posterior}
\citep{friel:pettitt:2008}, with a shrinkage prior such as the lasso
\citep{tibsh:1996}.  We then show how inference may proceed by
employing two (heretofore unrelated) data augmentation schemes: one
for the powered-up logistic likelihood; and the other for the
prior. The combined effect is a fully Gibbs MCMC sampler which, among
other advantages, allows estimators previously requiring custom
algorithms to be calculated via a single simulated annealing
\citep{kirkp:etal:1984} scheme.

Specifically, consider a set of binary responses, $y_i$, encoded as
$\pm 1$, regressed on $p$-dimensional predictors $x_i$ via the
logistic model $\mP(y_i = \pm 1| x_i, \beta) = (1 + e^{-y_i x_i^\top
  \beta})^{-1}$, for $i=1,\dots, n$.  When $p$ is large it is
paramount to infer $\beta$ under regularization or penalization.  A
common formulation
\citep[e.g.,][]{genkin:lewis:madigan:2007,park:hastie:2008} involves
finding regularized point-estimates $\hat{\beta}$ under an
$L_\alpha$-norm penalty, where parameters $\sigma^2 = (\sigma_1^2,
\dots, \sigma_p^2)$ control the relative penalization applied to each
predictor
\begin{equation}
\hat{\beta} = {\rm argmin}_\beta \; \; 
\sum_{i=1}^n \ln \left ( 1 + e^{ - y_i x_i^\top \beta } \right ) +
\nu^{-\alpha} \sum_{j=1}^p \left| \frac{\beta_j}{\sigma_j}
\right|^\alpha,
\;\;\;\;\; \alpha > 0.
\label{eq:pp}
\end{equation}
The parameter $\nu$ dictates the amount of regularization, the
relative pull ($\nu^{-1}$) or shrinkage of the $\beta_j$'s towards
zero.  Depending on the choice of $\alpha$, a number of algorithms
have been proposed to solve for $\hat{\beta}$.  For example,
\cite{madigan:ridgeway:2004} discuss how the LARS algorithm can be
useful as a subroutine for the popular case of $\alpha =
1$. 
It is typical to work with $x_i$ pre-scaled to have unit $L_2$-norm
with $\sigma_1 = \cdots \sigma_p = 1$ so that inference for $\beta$ is
equivariant under a re-scaling of the covariates. We follow this
convention in application but develop much of the discussion in the
general case for completeness.  
The special setting $\sigma_j^2 = \infty$ indicates no shrinkage for
$\beta_j$.  At least $\max\{0, p-n\}$ of the $\sigma_j^2$'s must be
finite to obtain stable estimators.  If there is an intercept in the
model, denoted by $\beta_0$, then it is common practice to absolve it
of penalty by taking $\sigma_0^2 = \infty$.  Throughout we begin the
$j$--indexing at $j=1$, ignoring the $0^{\mathrm{th}}$ term for
simplicity.

Our approach offers a fully probabilistic alternative by viewing the
objective function (\ref{eq:pp}) as a (log) posterior distribution
whose maximum {\em a posteriori} (MAP) estimator coincides with
$\hat{\beta}$. A multiplicity parameter $\kappa$ can then be
introduced to help find the MAP via simulation.  Our key insight,
which makes the simulation efficient, is that the logistic likelihood
component of the posterior can be written hierarchically using
$z$--distributions \citep{bn:kent:sorensen:1982}, leading to a data
augmentation scheme that generalizes that of \cite{holmes:held:2006}
[HH hereafter].  Combining this with a standard data augmentation for
the prior yields a highly blocked Gibbs MCMC algorithm for logistic
regression.  $Z$-distributions also suggest a new representation of
the likelihood that is equivalent (to HH) but requires $n$ fewer
latent variables.  Finally, we recognize that $\kappa$ has a secondary
use for binomial data (multiple $y$ observed for each $x$) which
otherwise would require more latent variables.

A distinctive feature of our framework is how it deals with the amount
of regularization, $\nu$, which is traditionally chosen by cross
validation (CV).  As an alternative, we may extend the hierarchical
model to include a prior for $\nu$ so that the marginal likelihood can
be computed and used to set $\nu = \hat{\nu}$, or to integrate $\nu$
out.  Posterior expectations, thus obtained, can give superior
point--estimators for $\beta$ in large-$p$ {\em linear} regression
contexts \citep{hans:2009}, and we show how this extends to logistic
regression.

The rest of the paper is outlined as follows. Section
\ref{sec:reglogit} provides our data augmentation strategies for
sparse high dimensional logistic regression, and Section
\ref{sec:inference} develops an MCMC scheme for estimation.  Section
\ref{sec:app} illustrates our approach with empirical comparisons to
modern competitors.  Finally, Section \ref{sec:discuss} concludes with
simple extensions and directions for future research.  An supporting
{\sf R} package, {\tt reglogit}, is available on CRAN.

\section{Regularized logistic regression via power-posteriors}
\label{sec:reglogit}

The central problem is to find the MLE, MAP, or posterior mean
estimator in logistic regression.  To do this, consider the following
power-posterior distribution inspired by Eq.~(\ref{eq:pp}):
\begin{equation}
\pi_{\kappa,\alpha} ( \beta | y, \nu, \sigma^2 ) = 
C_{\kappa, \alpha}(\nu) \exp \left \{ - \kappa \left ( 
\sum_{i=1}^n \ln \left ( 1 + e^{ - y_i x_i^\top \beta } \right ) +
\nu^{-\alpha} \sum_{j=1}^p \left| \frac{\beta_j}{\sigma_j}
\right|^\alpha \right ) \right \}.
\label{eq:pp2}
\end{equation}
The placement of $\kappa$ and $\alpha$ as subscripts in $\pi_{\kappa,
  \alpha}$ and $C_{\kappa, \alpha}(\nu)$, a normalization factor,
signals that these are user specified, not parameters to be estimated.
The $\alpha$ setting indicates the type of $L_\alpha$ regularization,
e.g., $L_1$ for absolute, and $L_2$ for quadratic.  The {\em
  multiplicity} (or thermodynamic) parameter $\kappa$, is a tool
borrowed from the power-posterior and simulated annealing literature
\citep[see,
e.g.,][]{pincus:1968,kirkp:etal:1984,doucet:godsill:robert:2002,jacquier:johannes:polson:2007,friel:pettitt:2008},
that facilitates several types of simulation based inference, as we
shall describe.

Power-posterior analysis can be helpful for calculating modes and
posterior means from complex optimization criteria, and marginal
likelihoods for Bayesian estimators.  Larger values of $\kappa$ cause
the density to concentrate near the modes, whereas small $\kappa$
distributes it away from the modes, in the troughs.  This motivates
two types of estimator.  First, $\mathbb{E}_{\kappa,\alpha}\{ \beta |
y, \sigma^2, \nu\}$ can be estimated for choices of $\nu$ by allowing
$\kappa$ to vary as in simulated annealing.  When $\nu=0$ the
estimator converges to the MLE as $\kappa \rightarrow \infty$.  When
$\nu > 0$, it converges to a posterior mode, or equivalently the
regularized estimator, $\hat{\beta}$ solving Eq.~(\ref{eq:pp}).
Furthermore, setting $\kappa = 1$ yields the posterior mean estimator.
Second, we recognize that $\kappa$ can be used to obtain an efficient
computational framework for binomial regression, where multiple binary
responses are recorded for each predictor.  
In what immediately follows, we regard $\kappa$ as fixed---a further
discussion is deferred to Section \ref{sec:inference}.

Observe that the likelihood--prior combination below yields
Eq.~(\ref{eq:pp2}) via Bayes' rule.
\begin{align}
L_\kappa( y | \beta )  &= e^{ - \kappa \sum_{i=1}^n \ln \left ( 1 +
      e^{ - y_i x_i^\top \beta } \right ) } = \prod_{i=1}^n \left (
    1 + e^{ - y_i x_i^\top \beta } \right )^{- \kappa } \label{eq:psl}
  \\
p_{\kappa, \alpha}( \beta | \nu , \sigma^2 )  &\propto \exp \left ( -
    \kappa \nu^{-\alpha} \sum_{j=1}^p | \beta_j/\sigma_j
    |^{\alpha} \right ) = \prod_{j=1}^p 
  \exp\left\{-\kappa \left|\frac{\beta_j}{\nu\sigma_j}
    \right|^\alpha\right\}. \nonumber
\end{align}
The following subsections provide data augmentation schemes for this
likelihood and prior.  They primarily concentrate on the $\alpha = 1$
case, i.e., the double--exponential or lasso prior, although results
are developed in generality when
possible. 
Section \ref{sec:discuss} briefly touches on the simpler $\alpha = 2$
case. 


\subsection{Hierarchical representation of the logistic}
\label{sec:hierlogit}

Extending a well-known technique for generating logistic regression
\citep[e.g.,][]{andrews:mallows:1974,holmes:held:2006}, we represent
the powered-up likelihood (\ref{eq:psl}) for $\beta$ as a marginal
quantity obtained after integrating over latent variables
$(z,\lambda)$, where $z = ( z_1, \ldots, z_n)$ and $\lambda = (
\lambda_1,\ldots,\lambda_n)$.  That is, each element of the product of
independent logistic terms can be written as a two-dimensional
integral:
\begin{equation}
L_\kappa(y|\beta) = \prod_{i=1}^n \int_0^{\infty} \int_0^\infty
p_\kappa(z_i | \beta, \lambda_i, y_i) p_\kappa(\lambda_i) \,
d\lambda_i \,dz_i. \label{eq:margl}
\end{equation}
This suggests a hierarchical representation in terms of latent
variables, $z_i$ for each $y_i$, mixed over $\lambda_i$.  It remains
to determine the appropriate form of $p_\kappa(z_i | \beta, \lambda_i,
y_i)$ and $p_\kappa(\lambda_i)$ so that $(1 + e^{-y_i x_i^\top
  \beta})^{-\kappa} = \int\!\int p_\kappa(z_i|\beta, \lambda_i, y_i)
p_\kappa(\lambda_i)\, d\lambda_i \, dz_i$.\footnote{The notation
  reserves $\pi(\cdot)$ for the marginal posterior $\beta$ as a visual
  queue for the quantity of primary interest.  All other probability
  densities use $p(\cdot)$, including the joint for latent
  $(z,\lambda)$ and all priors.}

Our key result, generalizing HH, relies on a scale mixture
representation of $z$--distributions \citep{bn:kent:sorensen:1982}.
These are characterized by their pdf as:
\begin{align}
Z(z; a, b, \sigma, \mu ) &
\equiv f_Z ( z | a , b, \sigma, \mu) = \frac{1}{\sigma B(a, b)} \frac{
  e^{a (z-\mu)/\sigma} }{ ( 1 + e^{(z-\mu)/\sigma} )^{ a +
    b} }  \label{eq:znorm} \\
&= \int_0^\infty \frac{1}{\sqrt{2 \pi \lambda\sigma^2}} 
\exp\left\{ - \frac{1}{2 \lambda \sigma^2} 
\left(z - \mu - \frac{1}{2}(a - b) \lambda \sigma\right)^2 \right\} 
 q_{ a , b } ( \lambda ) \,d \lambda \nonumber
\end{align}
where $q_{a, b}(\lambda)$ is a Polya distribution, i.e.,
an infinite sum of exponentials:
\begin{equation}
q_{a,b}(\lambda) = 
\sum_{k=0}^\infty w_k e^{ - \frac{1}{2} \psi_k \lambda } 
\;\;\;\;\; \mbox{where } \;\;\;
\psi_k = (a + k) (b + k), \label{eq:psi}
\end{equation}
and the weights $w_k$ are determined via $\delta = (a+b)/2$ and
$\theta = (a-b)/2$ as
\begin{equation}
  w_k  =  {-2\delta \choose k}
  \frac{ ( \delta + k )}{
    B( \delta+\theta , \delta - \theta )} 
  = \frac{(-1)^k (2 \delta ) \ldots ( 2
    \delta + k -1 )}{k !} 
  \frac{ ( \delta + k )}{B( \delta+\theta , \delta - \theta )}. \label{eq:w}
\end{equation}
This prior has a simple generative form:
\begin{equation}
\lambda \stackrel{D}{=} \sum_{k=0}^\infty 2 \psi_k^{-1} \epsilon_k, 
\;\;\;\;\; \mbox{where} \;\;\;
 \epsilon_k \sim \mathrm{Exp} (1). \label{eq:lamgen}
\end{equation}
Then, each component $(1 + e^{y x ^\top\beta })^{-\kappa}$ of the
likelihood (dropping $i$ subscripts) can be written as the cumulative
distribution function (cdf) evaluation (at zero) of a particular
$z$--distribution.  

\begin{theorem}
  The (powered up) logistic
  function may be represented as follows.
\begin{equation}
\left(\! 1 + e^{ -y x^\top \beta } \!\right)^{-\kappa }
\!= \int_0^{ \infty} \!\! \int_0^\infty   \frac{1}{\sqrt{2 \pi \lambda}}
 \exp \left \{  - \frac{1}{2 \lambda} 
   \left (z - y x^\top \beta - \frac{1}{2} 
     (1-\kappa) \lambda \right )^2 \right \}
 q_{1,\kappa}( \lambda ) \, d \lambda d z.
\label{eq:int}
\end{equation}
\label{thm:lhier}
\end{theorem}
\vspace{-1.25cm}
\begin{proof}
  If $z \sim Z(1, \kappa, 1, yx^\top \beta)$, then $F_Z(z) = 1 - (1 +
  e^{z - yx^\top \beta})^{-\kappa}$, giving $1 - F_Z(0) = (1 +
  e^{-yx^\top \beta})^{-\kappa}$.  In other words,
\begin{equation}
\left(1 + e^{ -y x^\top \beta } \right )^{-\kappa}
= \int_0^{ \infty} 
Z(z; 1, \kappa, 1, yx^\top \beta) \,dz,
\label{eq:int1}
\end{equation}
establishing the outer integration, over $z$, in Eq.~(\ref{eq:int}).
Applying the representation in Eq.~(\ref{eq:znorm}) yields the desired
result.
\end{proof}

The statistical implication of this is a hierarchical model which we
summarize in the following corollary.

\begin{corollary}
  The conditional distribution $p_\kappa(z_i | \beta, \lambda_i, y_i)$
  and the mixing distribution $q_{1,\kappa}(\lambda_i)$ imply that the
  latent $z_i$ follow
\begin{equation}
p_\kappa(z_i | \beta, \lambda_i, y_i) \equiv \mathcal{N}^+ \!\left( y_i x_i^\top
  \beta + \frac{1}{2} (1-\kappa) \lambda_i, \lambda_i \right),
\label{eq:zlat}
\end{equation}
where $\mathcal{N}^+$ is the normal distribution truncated to the
positive real line.
\label{cor:z}
\end{corollary}

In more compact notation, $z|\beta, \lambda, y \sim
\mathcal{N}_n^+((y.X)\beta + \frac{1}{2}(1-\kappa)\lambda, \Lambda)$,
where $y = (y_1, \dots, y_n)^\top$, $y.X = \mathrm{diag}(y) X$,
$\Lambda = \mathrm{diag}(\lambda_1,\dots, \lambda_n)$, and the
truncation is to the all-positive orthant.  Observe that, when $\kappa
= 1$, the above formulation is identical to the generative model
described by HH.  Given predictors $x_i$ and regression coefficients
$\beta$, generate $y_i \in \{-1,+1\} $ as
\begin{align}
  y_i &= \mathrm{sign}(z_i), & \mbox{where} && z_i &\sim
  \mathcal{N}(x_i^\top \beta, \lambda_i) & \mbox{and} && \lambda_i
  &=\sum_{k=1}^\infty \frac{2}{(1+k)^2} \epsilon_k, \;\;\; \epsilon_k
  \stackrel{\mathrm{iid}}{\sim} \mathrm{Exp}(1).
\label{eq:hhlik}
\end{align}
When $\kappa > 1$, the asymmetry of the $z$--distribution makes it
harder to extract $y_i$ from $y_i x_i^\top \beta +
\frac{1}{2}(1-\kappa)\lambda_i$, the mean of the truncated normal in
Eq.~(\ref{eq:zlat}).  In Section \ref{sec:binom}, we indirectly
suggest that one can interpret $\kappa y_i$ as a binomial response
when $\kappa$ is an integer.

\subsubsection*{An alternative $z$--representation:}

Theorem \ref{thm:lhier} shows how components of the powered-up
logistic likelihood can be represented hierarchically by the cdf of
$z$--distributions. We therefore call that multiplicity extension to
HH the {\em cdf representation}.  However, further inspection reveals
that it is possible to eliminate an integral in Eq.~(\ref{eq:margl})
and thus $n$ latent variables, and use the representation
\begin{align*}
(1 + e^{z_i - \mu_i})^{-\kappa} &\equiv Z(z_i; a=0, b=\kappa, 1, \mu_i) \\
&= \int_0^\infty \frac{1}{\sqrt{2 \pi \lambda_i}} 
\exp\left\{ - \frac{1}{2 \lambda_i} 
\left(z_i - \mu_i + \frac{1}{2}\kappa \lambda_i\right)^2 \right\} 
 q_{0,\kappa } ( \lambda_i ) \,d \lambda_i
\end{align*}
which avoids integrating over $z_i$.  Instead, set them to zero (and
$\mu_i = y_ix_i^\top \beta$) and directly obtain $(1 + e^{y_ix_i^\top
  \beta})^{-\kappa}$.  By analogy, we call this a {\em pdf
  representation} as it involves evaluating a particular $z$-density
function.
This simple representation is problematic, however, since the Polya
mixing density $q_{0,\kappa}$ is improper.  In particular, note that
$\psi_0 = 0$, resulting in a infinite weight in the generative
formulation (\ref{eq:lamgen}).

Fortunately, a similar representation
may be generated
\begin{align} 
(1 + e^{-\mu})^{-\kappa} &\equiv  Z(z; a, b, 1, \mu)\Big{|}_{z=0}
\nonumber \\
&= e^{a\mu} \int_0^\infty
\frac{1}{\sqrt{2\pi\lambda}} \exp \left\{-\frac{1}{2\lambda}\left(-\mu -
    \frac{1}{2}(a-b)\lambda\right)^2 \right\} q_{a,b}(\lambda) \,
d\lambda
\label{eq:pdf}
\end{align}
which involves a proper Polya mixing density as long as $(a,b) > 0$ and $a+b
= \kappa$.  In Section \ref{sec:conds} \& \ref{sec:app}, we show how
the extra $e^{a\mu}\equiv e^{a y_i x_i^\top \beta}$ poses no problem
for efficient inference, and that $(a=\frac{1}{2}, b=\kappa -
\frac{1}{2})$ works well in practice. But first, we complete the
power-posterior specification with a family of regularization priors
on $\beta$.

\subsection{Prior regularization}
\label{sec:hierlaplace}

Regularization is achieved via a family of priors,
$p_{\kappa,\alpha}(\beta | \nu, \sigma^2)$, implementing the
$L_\alpha$-norm 
via the decomposition $ p_{\kappa,\alpha}( \beta_j | \nu , \sigma^2 )
= \int p_{\kappa, \alpha}( \beta_j | \omega_j , \nu , \sigma^2)
p_\alpha( \omega_j )\, d \omega_j$, following \cite{carlin:polson:1991} and
\cite{park:casella:2008} 
in regularized (Bayesian) linear regression context.  The
idea is that, given $\beta_j = \frac{\nu}{\kappa^{1/\alpha}}
\sigma_j\sqrt{\omega_j} \epsilon_j$ and $\epsilon_j
\stackrel{\mathrm{iid}}{\sim} \mathcal{N}(0,1)$, small $\nu$ (i.e.,
heavy regularization) and large $\kappa$ (i.e., heavy concentration of
power-posterior density around the mode at the origin) {\em both}
shrink $\beta_j$ towards zero. We provide $p_\alpha( \omega )$
yielding the desired regularization penalty which, after unpacking
factors from $C_{\kappa, \alpha}(\nu)$ in Eq.~(\ref{eq:psl}), is
\begin{equation}
  \label{eq:prior}
  p_{\kappa, \alpha}(\beta|\nu, \sigma^2)  = \prod_{j=1}^p
  p_{\kappa, \alpha} (\beta_j|\nu, \sigma_j^2)  
  \propto \nu^{-p\kappa}
  \exp \left ( - \kappa \sum_{j=1}^p 
    \left|\frac{\beta_j}{\nu\sigma_j}\right|^\alpha  \right).
\end{equation}
\cite{box:tiao:1973} provide a general discussion of (\ref{eq:prior})
in the linear regression context.  Some notable special cases in the
recent literature on sparse logistic regression include the following:
when $\nu = 1$, $\alpha = 1$, and $\sigma_j = \lambda_j$ it is the
Laplace prior used in \cite{genkin:lewis:madigan:2007}; when $\alpha =
2, \sigma_j = 1$ and $\nu = \sigma^2$ it is the Gaussian prior, and
when $\alpha = 2, \sigma_j = 1$ and $\nu^{-1} = \lambda$ it is the
Laplace prior from \cite{krish:etal:2005}.\footnote{The $\lambda_j$
  and $\lambda$ variables correspond to the shrinkage parameters so
  named in our references.  They should not be confused with the
  latent $\lambda_i$ used in our hierarchical likelihood
  representation.}  Inference for $\nu$ in these cases typically
proceeds by CV, or by inspecting the paths of $\hat{\beta}_\nu$
solutions for varying $\nu$.  Assessing the uncertainty in estimators
$\hat{\beta}_{\hat{\nu}}$ on the final choice of $\hat{\nu}$ can pose
difficulties. 

Power posterior analysis offers an intriguing, third, option by
providing the potential for tractable marginalization over prior
uncertainty $\nu \sim p_{\kappa,
  \alpha}(\nu)$. 
Two particular choices in the $\alpha = 1$ case lead to efficient
inference by Gibbs sampling [Section \ref{sec:conds}].  One option is
an inverse gamma (IG) prior for $\nu^2$ with shape $r_\kappa = \kappa
(r+1) - 1$ and scale $d_\kappa = \kappa d$, where $\kappa=1$ yields a
base case $\mathrm{IG}(\nu^2; r, d)$ prior.  The second option is IG for
$\nu$, with identical powering-up identities.  It has lighter tails in
$\nu^{-1}$, thus providing more aggressive shrinkage.  

The prior in Eq.~(\ref{eq:prior})---for the purposes of efficient
inference [Section \ref{sec:inference}]---is an adaptation of a scale
mixture of normals result from \cite{west:1987} to account for
$\kappa$.  Specifically,
  \begin{equation}
    \label{eq:alpha-cdlike}
  p_{\kappa, \alpha}(\beta_j | \nu, \sigma_j^2)   = 
  \int_{\mathbb{R}_+} \mathcal{N}\!\left(\beta_j;
    0, \omega_j \cdot \frac{\nu^2  \sigma_j^{2}}{\kappa^{2/\alpha}}\right) 
    p_\alpha(\omega_j) \ d\omega_j,
  \end{equation}
  where $p_\alpha(\omega_j) \propto \omega_j^{-\frac{3}{2}}
  \mathrm{St}_{ \frac{\alpha}{2} }(\omega_j^{-1} )$ and
  $\mathrm{St}_{\alpha/2 }^+$ is the density function of a positive
  stable random variable of index $\alpha/2$.  In compact notation,
  $\beta|\sigma^2, \omega, \nu, \kappa \sim \mathcal{N}_p(0,
  \nu^2/\kappa^{2/\alpha} \Sigma \Omega)$ where $\Sigma =
  \mathrm{diag}(\sigma_1^2, \dots, \sigma_p^2)$ and $\Omega =
  \mathrm{diag}(\omega_1, \dots, \omega_p)$.  An important corollary,
  obtained by adapting an \cite{andrews:mallows:1974} result, is that
  if $\alpha = 1$, $\omega_j \stackrel{\mathrm{iid}}{\sim}
  \mathrm{Exp}(2)$, and $\sigma_j = 1$ for $j=1,\dots,p$ then
  $p_\kappa(\beta | \nu)$ is double exponential (Laplace) with a mean
  zero and scale $\nu^2/\kappa^2$.

\section{Simulation-based logistic regression}
\label{sec:inference}

We develop a Gibbs sampling algorithm [Section
\ref{sec:anneal}] for sampling the augmented power-posterior $p_\kappa(\beta,
z, \omega, \lambda, \nu| y, \sigma^2)$, for any $\kappa$.  We first
derive the relevant posterior conditionals [Section \ref{sec:conds}],
treating cdf and pdf representations in turn.  When $\kappa = 1$ the
marginal samples of $\beta$ summarize the posterior distribution of
the main parameters of interest. Obtaining the MAP or MLE requires an
inhomogeneous Markov chain [Section \ref{sec:anneal}].  Finally, we
describe how a vectorized $\kappa$ can facilitate efficient Bayesian
binomial regression [Section \ref{sec:binom}].

\subsection{Posterior conditionals}
\label{sec:conds}

To begin, consider the latent $z$ and $\lambda$ variables in the cdf
and pdf representations, in turn, followed by the coefficients $\beta$
and corresponding regularization prior parameters $(\omega, \nu)$.

\subsubsection*{Latent likelihood parameters $(z,\lambda)$}

By construction [Eq.~(\ref{eq:zlat}) of Corollary \ref{cor:z}], the
posterior full conditional for the latents, $p_\kappa(z_i| \beta, \lambda_i,
y_i)$, is a truncated (non-negative) normal distribution.  Obtaining
samples, independently for $i=1,\dots,n$, is straightforward following
the methods of \cite{robert:1995}.

Sampling from the full conditional $p_\kappa(\lambda_i| \beta, z_i,
y_i)$ is complicated by the infinite sum in the expression for the
prior (\ref{eq:psi}), which precludes a na\"ive approach via
truncation since certain combinations of $\lambda_i$ and $b \equiv
\kappa$ can give highly inaccurate, even negative, evaluations.  HH
derive an expression for this conditional when $\kappa = 1$ and
provide a rejection sampling algorithm by squeezing
\citep{devroye:1986}.  Although adaptable for general $\kappa$, we
prefer a Rao--Blackwellized approach.  Interchanging the order of
integration in Eq.~(\ref{eq:int}) suggests a corollary to Theorem
\ref{thm:lhier} that is helpful in constructing a Metropolis--Hastings
(MH) scheme for obtaining $\lambda_i$ draws.

\begin{corollary}
  The following is an alternate integral representation of the
  logistic function
\[
\exp \left \{ - \kappa 
\ln \left ( 1 + e^{ - y_i x_i^\top \beta } \right ) \right \}
 = \int_0^\infty  \Phi \left( 
\frac{ - y_i x_i^\top \beta - \frac{1}{2} ( 1 - \kappa) 
\lambda_i }{\sqrt{\lambda_i}} \right )
 q_{1,\kappa}(\lambda_i )d \lambda_i,
\]
where $\Phi$ is the cdf of the standard normal distribution.
\end{corollary}

\noindent Proposals $\lambda_i' \sim q_{1,\kappa}(\lambda)$ can then
be accepted via MH with probability $\min\{1,A_i\}$ where
\begin{equation}
A_i = \frac{ \Phi\{(- y_i x_i^\top \beta - \frac{1}{2} ( 1 - \kappa)
\lambda_i')/\sqrt{\lambda_i'}\}}{\Phi\{(- y_i x_i^\top \beta - \frac{1}{2} ( 1 - \kappa)
\lambda_i)/\sqrt{\lambda_i}\}}.
\label{eq:lama}
\end{equation}
%
Good proposals may be obtained by truncating the sum in
Eq.~(\ref{eq:lamgen}) at $K=100$ for $\kappa=b=1$, with improvements
for larger $\kappa$.  Direct sampling is also possible
\citep[e.g.,][]{weron:1996}.

Empirically, the MH acceptance rate is high ($>$ 90\%) for $\kappa =
1$ because posterior is similar to the prior ($q_{1,1}$).  Therefore
the MH scheme may be preferable to the rejection/squeezing method of
HH who report acceptance rates as low as 25\%.  Both rates decline as
$\kappa$ is increased, but the MH rate is still above $1\%$ for
$\kappa = 20$.  A good rule of thumb is to thin $\lceil \kappa \rceil$
draws for each draw saved, which is reasonable from a computational
standpoint as sampling from $q_{a,b}$ is fast.  Even when thinning
more than 10-fold, the MH sampler is competitive to HH/Devroye in
terms of sheer speed.  The MH requires two $\Phi$ evaluations, a few
arithmetic operations, and two square roots.  HH/Devroye, by contrast,
can perform dozens (or more) expensive operations such as {\tt pow}
before the squeeze is made.
Finally, drawing $\lambda_i$ unconditional on $z_i$ yields lower
autocorrelation in the overall joint MCMC sampling scheme.

The pdf representation is simpler since $z_i$ is set to zero. Proposed
$\lambda_i' \sim q_{a,b}$ may be accepted or rejected via MH by
exchanging a cdf for a pdf in Eq.~(\ref{eq:lama}) and replacing
$\frac{1}{2}(1-\kappa)$ with $\frac{1}{2}(a-b)$.  Another feature that
works well for the pdf representation is an adaptation of the slice
sampler of \citet{godsill:2000}.  Given $\lambda_i$, the next sample
$\lambda_i'$ may be obtained via an auxiliary uniform random variable
as follows.  Let $\phi_i \equiv \phi\{(-y_i x_i^\top \beta
+\frac{1}{2}(a-b)\lambda_i)/\sqrt{\lambda_i}\}$, where $\phi$ is the
pdf of a standard normal distribution.  Then sample
\begin{align}
u | \lambda_i, x_i, y_i, \beta &\sim 
U[0, \phi_i], && \mbox{followed by} &
\lambda_i'| u, x_i, y_i, \beta &\sim q_{a,b}(\lambda_i') 
\mathbb{I}_{\{\phi_i' > u\}}, \label{eq:slice}
\end{align}
where the second step is facilitated by accept/rejects following
random draws from the Polya mixing density.  Although more automatic
in that it does not require thinning, we show in Section \ref{sec:st}
that the MH scheme is faster overall.  The two methods behave
similarly when $\kappa$ gets large, causing the rate of
rejections/required thinning to increase.

\subsubsection*{Regularized regression coefficient parameters
 $(\beta,\omega,\nu)$}

In the cdf representation, the multivariate normal priors for $z$
[Section \ref{sec:hierlogit}] and $\beta$ [Section
\ref{sec:hierlaplace}] combine to give $\beta | z, \omega,
\lambda, \nu, \kappa \sim \mathcal{N}_p(\tilde{\beta}, V)$ with
hyperparameters
\begin{align*} 
&& \tilde{\beta} &= V (y.X)^\top \Lambda^{-1}\left(z -
  \frac{1}{2}(1-\kappa)\lambda\right), && \mbox{and} &
V^{-1} &= (\nu/\kappa^{1/\alpha})^{-2} \Sigma^{-1} \Omega^{-1}
+ (y.X)^\top \Lambda^{-1} (y.X).
\end{align*}
Obtaining $V$ from $V^{-1}$ is generally $O(p^3)$, which represents a
significant computational burden in the $p\gg n$ context. By employing
the Sherman--Morrison--Woodbury formula
\citep[e.g.,][pp.~67]{berns:2005}, it is possible to use an $O(n^3)$
operation instead, which could represent significant savings.  In the
pdf representation a similar combination of regularization penalties
and likelihoods gives an identical $V^{-1}$ expression, but a new
$\tilde{\beta} = (a - \frac{1}{2}[a-b]) V X^\top y$ [see Appendix
\ref{sec:beta}].  Choosing $(a=\frac{1}{2},b=\kappa - \frac{1}{2})$
gives $\tilde{\beta} = \frac{\kappa}{2} V X^\top y$, a particularly
simple expression that may be used for $\kappa > \frac{1}{2}$.  It is
interesting to observe that the parameters $(\lambda, \omega, \nu)$ only
enter into the conditional for $\beta$ through $V$ in the pdf
representation.

The full conditional distribution of each latent $\omega_j$ is
proportional to the integrand of Eq.~(\ref{eq:alpha-cdlike}). When
$\alpha = 1$ we have the following adaptation of a standard result.
\begin{corollary}
  \label{cor:omega-full-cond}
  For $\alpha=1$, the full conditional distribution of the reciprocal
  of $\omega_j^{-1}$ follows an inverse Gaussian distribution:
  $\omega_j^{-1} |\beta_j,\nu, \kappa \sim \ig (\frac{\nu}{\kappa}
  |\frac{\beta_j}{\sigma_j}|^{-1} , 1)$.
\end{corollary}

\begin{proof}
From the integrand in Eq.~\eqref{eq:alpha-cdlike} with $\alpha =
1$ we have
\begin{align*}
p_\kappa( \omega_j | \beta_j , \nu ) & \propto \frac{1}{ \sqrt{2 \pi \omega_j}} 
\exp\left\{ - \frac{1}{2} \left
    ( \frac{ \kappa^2 \beta_j^{2} }{\nu^2 \sigma_j^2 \omega_j} 
+ \omega_i \right ) \right\} 
 \equiv \gig\!\left (\omega_j;  \frac{1}{2} , 1 ,  \frac{\kappa^2
     \beta_j^2}{\nu^2 \sigma_j^2} \right ),
\end{align*}
which is implies that $\omega_j^{-1} \sim
\mathrm{IN}\left(\frac{\nu}{\kappa}\left|
    \frac{\beta_j}{\sigma_j}\right|, 1\right)$. [See Appendix
\ref{sec:gig} for IN/GIG definitions].  
\end{proof}

Our IG priors for $\nu$ are both conditionally conjugate. An IG prior
for $\nu^2$ and the representation in Eq.~(\ref{eq:alpha-cdlike}) gives
\begin{align*}
  \nu^2 |\beta, \omega, \kappa \sim \mathrm{IG}\left( r_\kappa +
    \frac{\kappa p}{2}, d_\kappa + \frac{\kappa^2}{2} \sum_{j=1}^p
    \frac{\beta_j^2}{\sigma_j^2 \omega_j}\right).  \intertext{An IG
    prior for $\nu$ leads to efficiency gains (in addition to better
    tail properties) since there is no conditioning on $\omega$.
    Using Eq.~(\ref{eq:prior}) directly in then gives} \nu
  |\beta, \kappa \sim \mathrm{IG}\left( r_\kappa + \kappa p, d_\kappa
    + \kappa \sum_{j=1}^p
    \left|\frac{\beta_j}{\sigma_j}\right|\right),
\end{align*}
extending the analysis of \cite{park:casella:2008}. 


\subsection{Gibbs sampling and annealing for point estimators}
\label{sec:anneal}

A full Gibbs sampling
algorithm for both cdf and pdf representations is outlined in Figure
\ref{f:gs}.
\begin{figure}
\small
\centering
\fbox{
\begin{minipage}{16cm}
Inputs:
\begin{itemize}
\item Data: $n\times p$ response-multiplied design matrix $y.X$
\item Settings: multiplicity $\kappa > 0$; scale factors $\sigma_1,
  \dots, \sigma_p$ where $\Sigma = \mathrm{diag}(\sigma_1^2, \dots, \sigma_p^2)$; representation type $R \in
  \{\mathrm{cdf}, \mathrm{pdf}\}$; Polya parameters $(a,b)$ where
  $(a=1,b=\kappa)$ if $R = \mathrm{cdf}$ or $(a,b) > 0$ and $a+b =
  \kappa$, otherwise; prior parameters $(r_\kappa, d_\kappa) > 0$;
  sample size $S$
\item Initial values: $\beta^{(0)} = (\beta_1^{(0)}, \dots,
  \beta_p^{(0)})^\top$, $\nu^{(0)}$, latents $\Lambda^{(0)} =
  \mathrm{diag}(\lambda_1^{(0)}, \dots, \lambda_n^{(0)})$, and if $R =
  \mathrm{cdf}$ also include latents $z^{(0)} = (z_1^{(0)}, \dots,
  z_n^{(0)})^\top$
\end{itemize}
Gibbs sampling, for iterations $s=1,\dots,S$:
\begin{enumerate}
\item For $j=1,\dots, p$ take 
$
\omega_j^{-1}
  \sim \ig \left(\frac{\nu^{(s-1)}}{\kappa}
  \left|\frac{\beta_j^{(s-1)}}{\sigma_j}\right|^{-1} , 1\right),
$
and let $\Omega^{(s)} = \mathrm{diag}(\omega_1^{(s)}, \dots, \omega_p^{(s)})$
\item For $i=1,\dots,n$ do the following depending on the
  representation $R$
  \begin{itemize}
    \item propose $\lambda_i' \sim q_{a,b}$ approximately via
      (\ref{eq:lamgen}) as  $\lambda_i' = \sum_{k=1}^K
      \frac{2 \epsilon_k}{(a+k)(b+k)}$, where $\epsilon_k
      \stackrel{\mathrm{iid}}{\sim} \mathrm{Exp}(1)$ and
      $K$ large
    \item draw $u \sim \mathrm{Unif}(0,1)$ and if $u < A_i$ where
\[
A_i = \left\{ \begin{array}{rl}
\frac{ \Phi\left\{(- y_i x_i^\top \beta^{(s-1)} - \frac{1}{2} ( 1 - \kappa)
\lambda_i')/\sqrt{\lambda_i'}\right\}}{\Phi\left\{(- y_i x_i^\top \beta^{(s-1)} - \frac{1}{2} ( 1 - \kappa)
\lambda_i^{(s-1)})/\sqrt{\lambda_i^{(s-1)}}\right\}} & \mbox{if } \;
R=\mathrm{cdf} \\
\frac{ \phi\left\{(- y_i x_i^\top \beta^{(s-1)} - \frac{1}{2} ( 1 - \kappa)
\lambda_i')/\sqrt{\lambda_i'}\right\}}{\phi\left\{(- y_i x_i^\top \beta^{(s-1)} - \frac{1}{2} ( a - b)
\lambda_i^{(s-1)})/\sqrt{\lambda_i^{(s-1)}}\right\}}  & \mbox{otherwise}
\end{array} \right.
\]
then take $\lambda_i^{(s)} = \lambda_i'$, or  take $\lambda_i^{(s)} =
      \lambda_i^{(s-1)}$ otherwise.
    \end{itemize} 
Then let $\Lambda^{(s)} = \mathrm{diag}(\lambda_1^{(s)}, \dots,
\lambda_p^{(s)})$ and  $\lambda^{(s)} = (\lambda_1^{(s)}, \dots,
\lambda_p^{(s)})^\top$

\item If $R = \mathrm{cdf}$ then for $i=1,\dots,n$ draw
$
z_i^{(s)} \sim \mathcal{N}^+ \!\left( y_i x_i^\top
  \beta^{(s-1)} + \frac{1}{2} (1-\kappa) \lambda_i^{(s)},
  \lambda_i^{(s)} \right)
$, and collect them as $z^{(s)} = (z_1^{(s)}, \dots, z_n^{(s)})^\top$
\item 
\begin{itemize}
\item Calculate $V^{-1(s)} = (\nu^{(s)}/\kappa)^{-2} \Sigma^{-1}
  \Omega^{-1(s)} + (y.X)^\top \Lambda^{-1(s)} (y.X)$. If $R =
  \mathrm{cdf}$ then calculate $\tilde{\beta} = V^{(s)} (y.X)^\top
  \Lambda^{-1(s)}\left(z^{(s)} -
    \frac{1}{2}(1-\kappa)\lambda^{(s)}\right)$, otherwise
  $\tilde{\beta}^{(s)} = (a - \frac{1}{2}[a-b]) V^{(s)} X^\top y$
\item  Draw $\beta^{(s)} \sim \mathcal{N}_p(\tilde{\beta}^{(s)}, V^{(s)})$ 
\end{itemize}
\item Draw $\nu^{(s)} \sim \mathrm{IG}\left( r_\kappa + \kappa p,
    d_\kappa + \kappa \sum_{j=1}^p
    \left|\frac{\beta_j^{(s)}}{\sigma_j}\right|\right)$
\end{enumerate}
Output: $\{\beta^{(s)}\}_{s=1}^S$, $\{\nu^{(s)}\}_{s=1}^S$, latents
$\{\lambda^{(s)}\}_{s=1}^S$, and if $R = \mathrm{cdf}$ also include latents
  $\{z^{(s)}\}_{s=1}^S$
\end{minipage}
}
\caption{Pseudocode for simulation based regularized logistic regression.}
\label{f:gs}
\end{figure}
For compactness, variations with slice sampling for $\lambda$ [in the
pdf case] or a prior on $\nu^2$ are not shown.  The former requires
replacing each iteration of step 2 by the method surrounding
Eq.~(\ref{eq:slice}).  The latter requires drawing $\nu^{2(s)} \sim
\mathrm{IG}\left( r_\kappa + \frac{\kappa p}{2}, d_\kappa +
  \frac{\kappa^2}{2} \sum_{j=1}^p \frac{\beta_j^{2(s)}}{\sigma_j^2
    \omega_j^{(s)}}\right)$ in step 5 and specification of
$\nu^{2(0)}$ on input.  Initial latent $\omega_j$ values are not
required.

The samples on output may be used to approximate expectations under
the power-posterior distribution with multiplicity $\kappa$.  If
$\kappa = 1$ then these are samples from a well-defined posterior
distribution which may be used, e.g., to approximate the posterior
mean of $\beta$ or provide samples from the posterior predictive
distribution.  Both take into account the full the uncertainties of all
parameters (including $\nu$) into account---a feature unique to full
Bayesian analysis.

Settings of $\kappa > 0$ are useful for finding other popular
estimators via {\em simulated annealing} (SA).
In our context, SA establishes an inhomogeneous Markov chain over a
sequence of power-posteriors, starting with $\kappa = 1$ and then
increasing according to a pre-determined schedule.  Except when Gibbs
sampling is possible for all $\kappa$ (as for our power-posterior), it
is usually difficult to ensure that the Markov chain mixes well,
particularly when $\kappa$ increases.  A pragmatic approach starts at
$\kappa \approx 1$, and systematically makes modest increases in
$\kappa$ until Monte Carlo variation in the power-posterior
expectations of the quantities of interest is below a pre-determined
threshold.  Each annealing iteration is initialized with the last
value $\beta^{(S)}$, $\nu^{(S)}$, $\lambda^{(S)}$ and $z^{(S)}$, from
the previous iteration, thereby stitching the inhomogeneous Markov
chains together.  The chain for each $\kappa$ must have enough
iterations to establish convergence to its particular power-posterior.

Annealed procedures such as ours present an MCMC alternative to
EM-style algorithms.  Importantly, SA is known to converge to the
global optima in certain conditions (when $\kappa \rightarrow
\infty$), whereas EM is only guaranteed to find a local optima.
Although convergence for EM is usually quick, there are no guarantees
that it will be so and indeed there are examples, particularly in high
dimensional settings, where convergence can be arbitrarily slow.  SA
however, comes with the burden of choosing the schedule for increasing
$\kappa$.  We have found that for our regularized logistic regression
scheme, convergence is fast and mixing so good that short schedules
such as $\kappa = 1, 5, 10, 20$ are a safe default [see Section
\ref{sec:app}].  Even jumping immediately to modest $\kappa\;(\approx
20$), skipping $\kappa=1$, can very often yield cheap and accurate
approximations.

But perhaps the most noteworthy difference between our simulation
approach and previous methods (like EM) are the myriad of options
(beyond CV) for inferring $\nu$.  One option, in the classical
context, is to use annealing to find the joint mode of $(\beta, \nu)$.
Another option is to first use samples from the posterior marginal
$p(\nu|X, y)$ to estimate the posterior mean $\hat{\nu} = \E\{\nu|X,
y\}$, and then proceed to estimate $\hat{\beta} = \mathbb{E}_\kappa \{
\beta|\hat{\nu}, X, y\}$ as before.  In Figure \ref{f:gs} this would
be facilitated by inputting $\nu^{(0)} = \hat{\nu}$ and replacing step
5 with $\nu^{(s)} = \nu^{(s-1)}$.  Our experience is that the former
works well for small $p$ problems, and the latter for large $p$.  When
$p$ is large, the joint prior for $(\beta, \nu)$ dominates near the
posterior mode of $\nu$, which tends to zero and yields $\hat{\beta} =
0$, which is not helpful.  The marginal posterior mean is far less
sensitive to the regularization prior, and represents a more
convenient choice for large $p$ applications. In Section
\ref{sec:spam} we provide an example where the joint mode is easy to
find with a few dozen predictors, whereas an interaction expanded
version using thousands requires more care.

\subsection{Efficient handling of binomial data}
\label{sec:binom}

Another advantage of our approach is the extension to binomial data,
where binary responses are collected repeatedly and independently,
$n_i$ times for subjects with the same covariates $x_i$.  Contingency
tables are one important example.  A typical (un-regularized) logistic
regression model is $y_i|x_i \sim \mathrm{Bin}(n_i, \mu_i)$, where $\mu_i
= e^{\eta_i}/(1 + e^{\eta_i})$ and $\eta_i$ is linear in $x_i$. One
way to situate such data within this article's regularized logistic
regression framework is to {\em flatten} it, so that $n_i$ components
appear in the likelihood for each subject $i$: $\prod_{j=1}^{n_i} (1 +
e^{-y_{ij} x_i^\top \beta})^\kappa$, using the binary encoding $y_{ij}
\in \{-1,1\}$ giving $|\sum_{j=1}^{n_i} y_{ij} | = n_i$.  This allows
inference to proceed as described in Section \ref{sec:inference}, but
it can lead to an inefficient MCMC scheme if the $n_i$ are large due
to the $n_i$ latents required for each $i$.  It turns out that it is
possible to use only two latents for each $i$, echoing a feature of
methods described by \citet{frueh2:etal:2009}.

Observe that the component of the likelihood for subject $i$ may be
equivalently written with just two terms as $(1 + e^{- x_i^\top
  \beta})^{\kappa y_i} (1 + e^{x_i^\top \beta})^{\kappa (n_i - y_i)}$,
which is proportional to the $i^\mathrm{th}$ component of a typical
binomial likelihood with logit link.  Hence the full likelihood,
with $m$ unique subjects, can be written as $\prod_{i=1}^m (1 + e^{-
  x_i^\top \beta})^{\kappa_{i+}} (1 + e^{x_i^\top
  \beta})^{\kappa_{i-}}$, where $\kappa_{i+} = \kappa y_i$ and
$\kappa_{i-} = \kappa (n_i - y_i)$.  This is identical to a
$z$--distribution representation of the logistic likelihood with $2m$
terms, which may be much less than the $\sum_{i=1}^m n_i$ produced by
flattening.  The first $m$ terms use response ``data'' $y_i' = +1$
with multiplicity parameter $\kappa_{i+}$, and the second $m$ terms
use $y_i' = -1$ with $\kappa_{i-}$.  A {\em multiplicity
  implementation} is therefore facilitated by forming vectors $y'$ and
$\kappa'$, each of length $n=2m$, and using $\prod_{i=1}^n (1 + e^{-
  y_i' x_i^\top \beta})^{\kappa_{i}'}$.

The MCMC scheme proceeds as in Section \ref{sec:inference} by
vectorization.  For example, steps 2 and 3 for $z_i$ and $\lambda_i$
would use $\kappa_i'$ instead of $\kappa$.  For $\beta$ in step 4 with
$(a=0.5, b=\kappa-0.5)$ say, replace $\kappa 1_n$ with the $\kappa'$
vector in the expression for $\tilde{\beta}$.  Terms can be eliminated
from the likelihood, thus eliminating the corresponding latents, where
$\kappa_i' = 0$, as is the case when $y_i \in \{0,n_i\}$.  The
original, scalar, $\kappa$ is used for the conditionals corresponding
to the parameters of the prior.  For example, the posterior
conditional covariance $V$ of $\beta$ is
unchanged.  

\section{Applications}
\label{sec:app}


\subsection{Pima Indian data}
\label{sec:pima}

The Pima Indian diabetes data [UCI Machine Learning Repository
\citep{Asuncion+Newman:2007}] includes outcomes for diabetes tests
performed on $n=768$ women of Pima heritage with 8 real-valued
predictors. Some of the predictors have many zeros, which may
reasonably be interpreted as ``missing'' values.  To remain consistent
with the treatment of this data by HH, and other authors, we do not
treat these values in any special way.  The following analysis
highlights properties of regularized estimators of $\beta = (\beta_0
\equiv \mu, \beta_1, \dots, \beta_8)$ obtained with $\alpha = 1$,
$\sigma_j = 1$ for $j=1,\dots,8$, and $T=1000$ samples from the
resulting posterior (the first 100 as burn-in).

\begin{figure}[ht!]
\centering
\includegraphics[scale=0.35,trim=5 30 20 0,clip=TRUE]{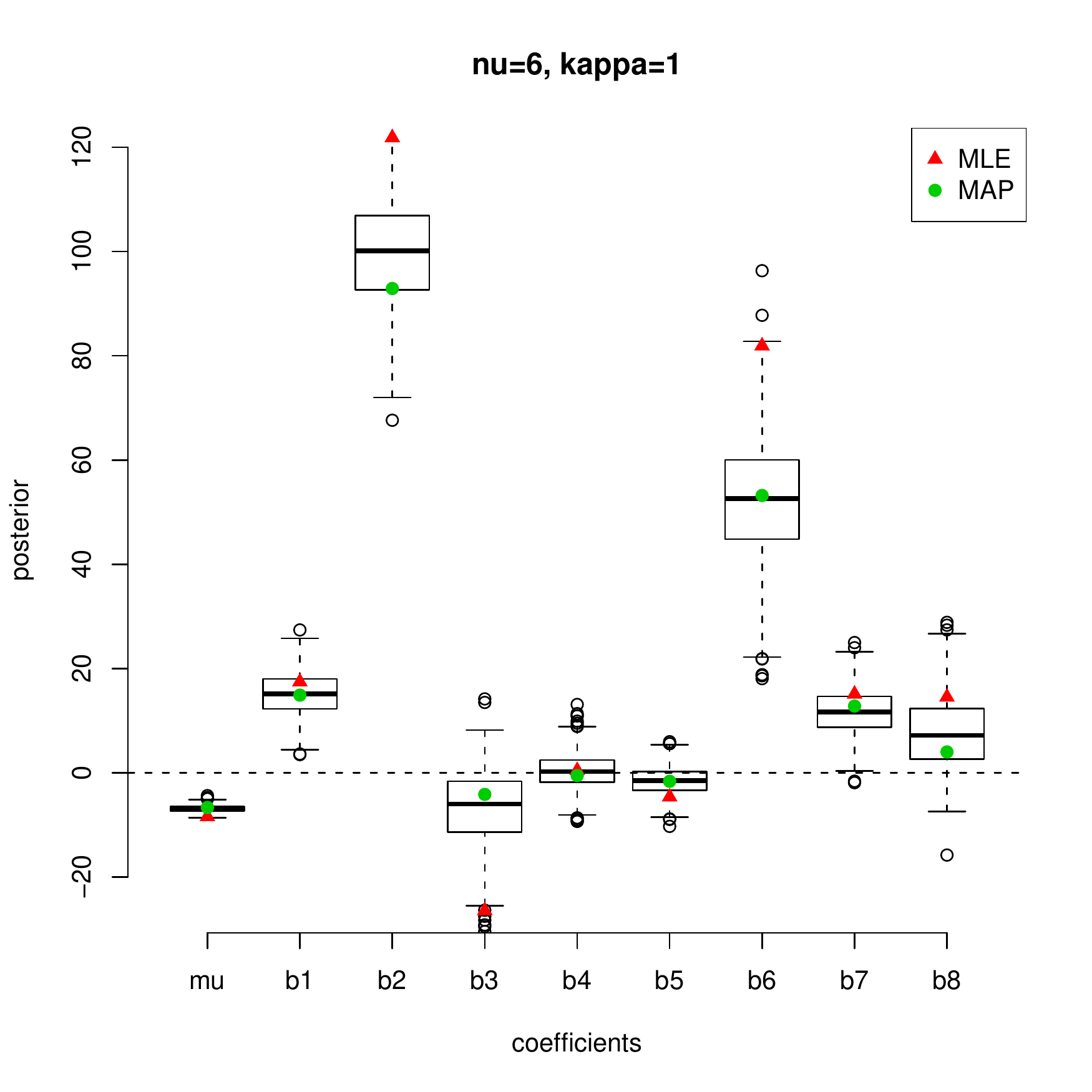}
\includegraphics[scale=0.35,trim=70 30 20 0,clip=TRUE]{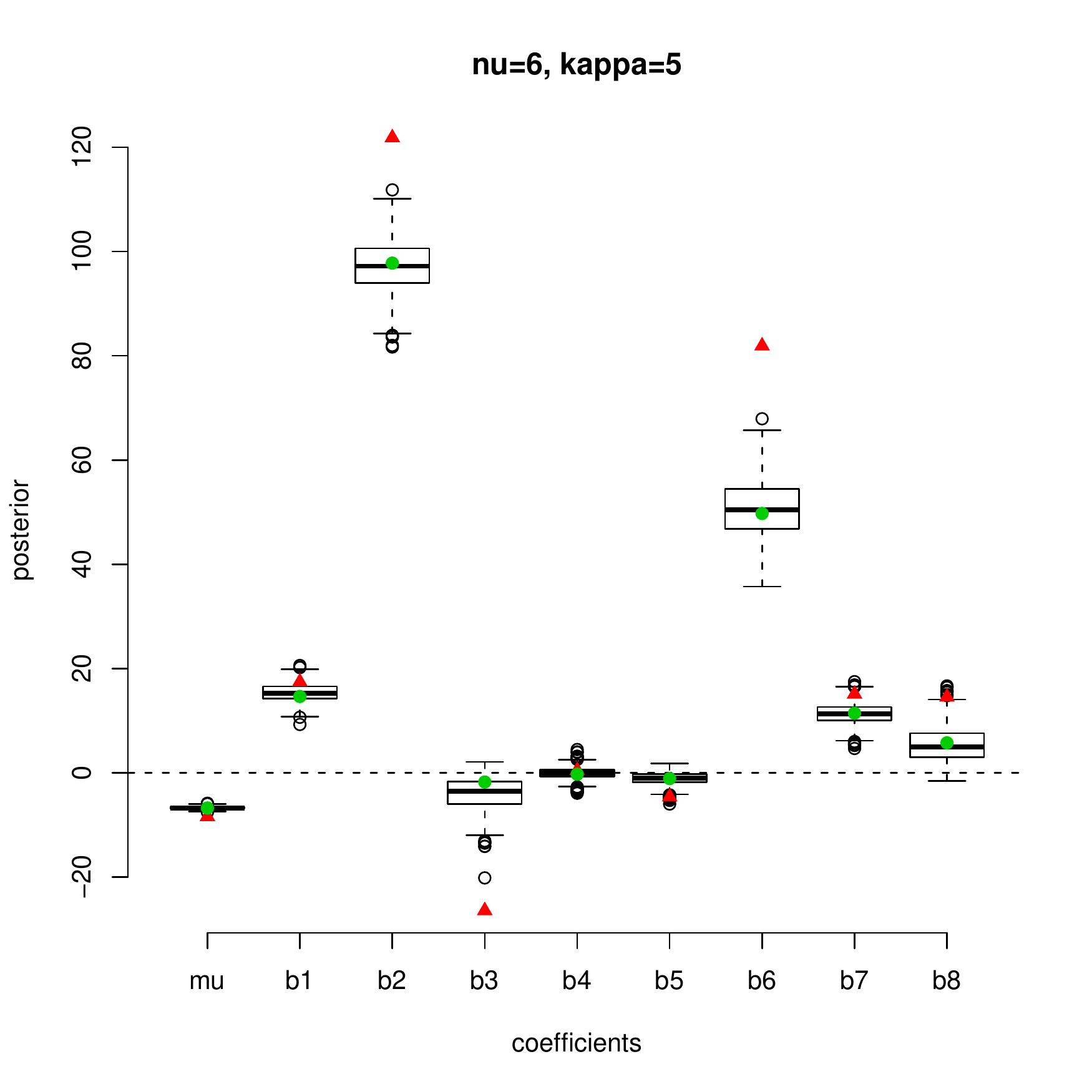}
\includegraphics[scale=0.35,trim=70 30 20 0,clip=TRUE]{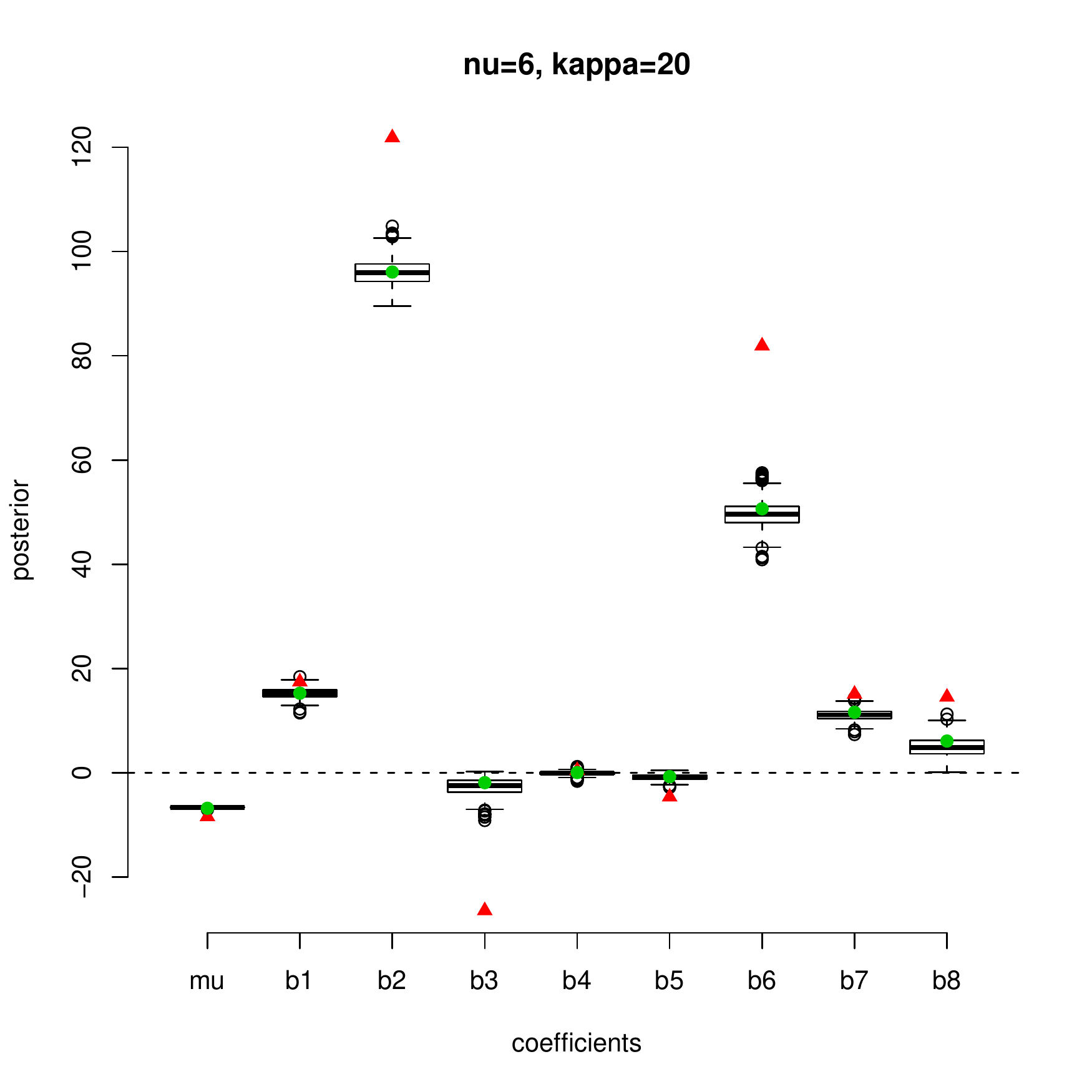}
\vspace{-0.25cm}
\caption{Power-posteriors for the Pima Indian data: $\nu = 6$, $\kappa
  \in \{1,5,20\}$.}
\label{f:pima:beta}
\end{figure}

Figure \ref{f:pima:beta} summarizes the marginal power posterior(s)
for $\beta$ with boxplots.  Three settings of $\kappa \in \{1,5,20\}$
(each panel) were used, and heavy regularization (fixing $\nu = 6$)
was applied.  Only the first panel ($\kappa = 1$) summarizes samples
from the true posterior.  The $\kappa > 1$ settings are useful for
obtaining other estimators.  The MLE, obtained from the {\tt glm}
command in {\sf R} \citep{cran:R}, and the MAP as estimated from the
sample(s), are also shown.  Shrinkage is apparent in the divergence
between the MAP and MLE values in all panels.  Observe how the
quartiles and outliers converge on the MAP as $\kappa$ is increased,
reflecting higher confidence in the accurate estimation of those
values.  Convergence is particularly rapid for the intercept term, and
the two coefficients with considerable mass near zero ($\beta_4$ and
$\beta_5$). These columns of $X$ have the highest concentration of
``missing'' values (30\% and 49\% respectively), so it is not
surprising the that MAP estimator excludes them.

\begin{figure}[ht!]
\centering
\includegraphics[scale=0.6,trim=15 0 15 40]{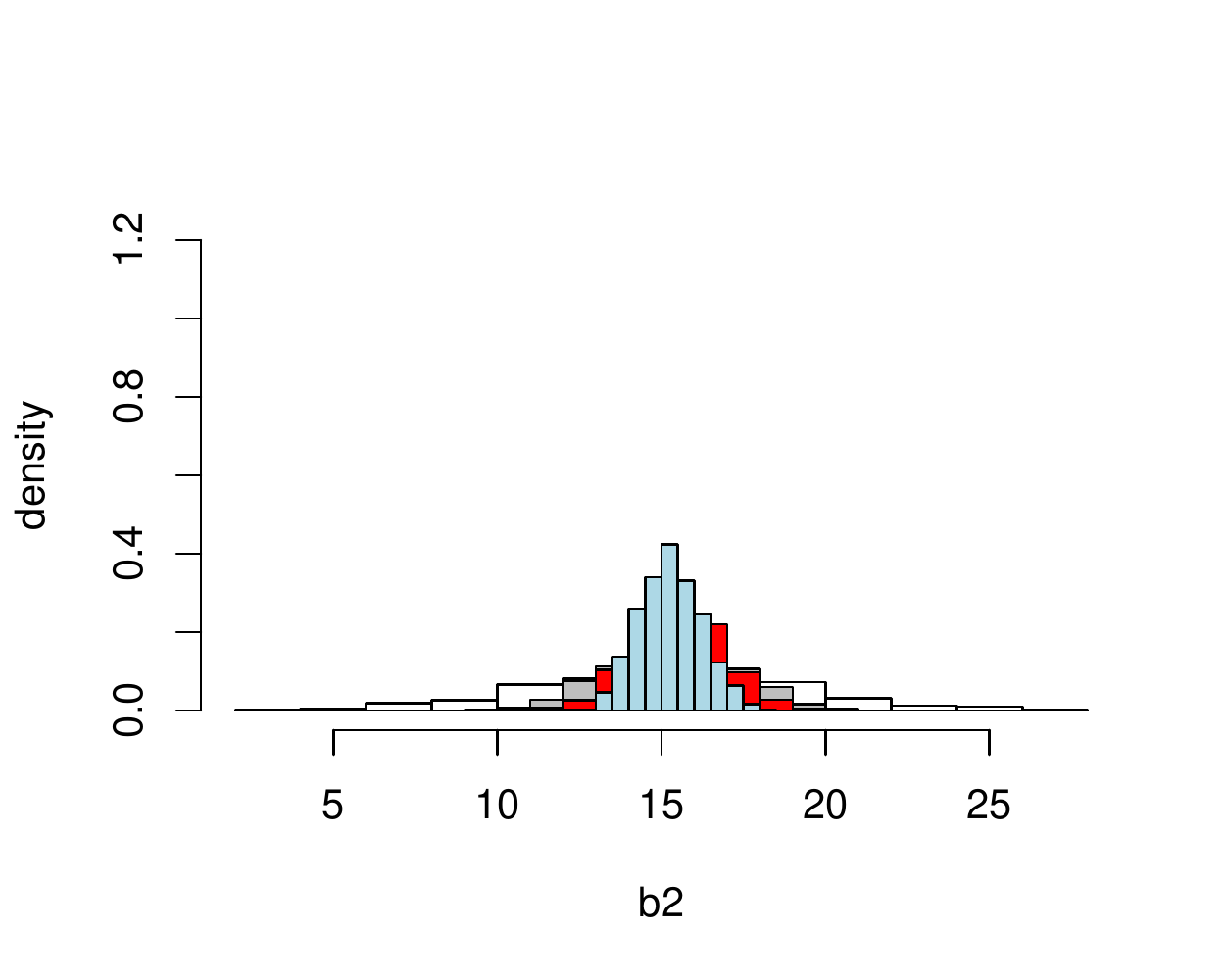}
\includegraphics[scale=0.6,trim=30 0 30 40, clip=TRUE]{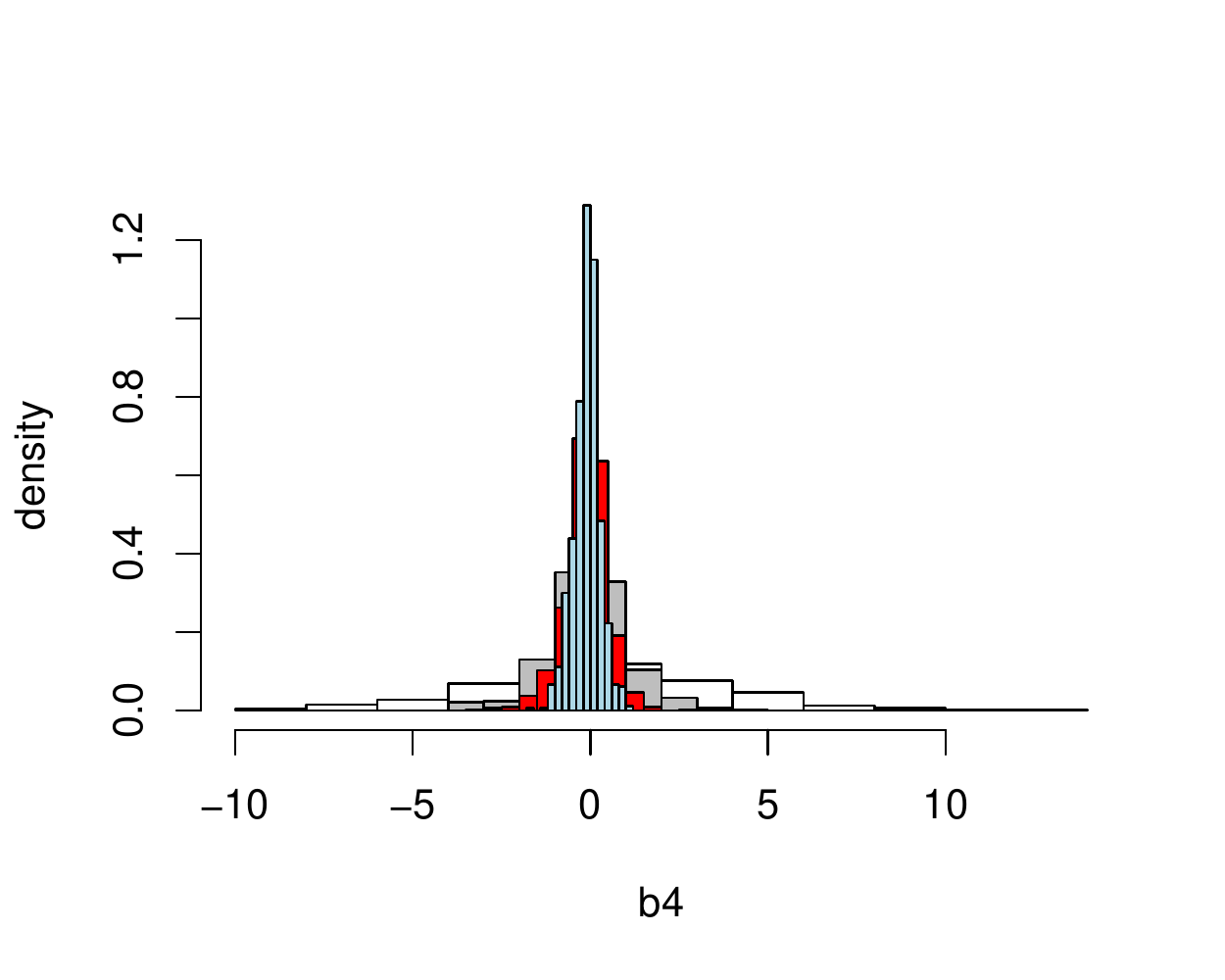}
\vspace{-0.5cm}
\caption{Illustrating the concentration of posterior mass of $\beta_1$
  and $\beta_4$ on the Pima Indian data for $\kappa \in
  \{1,5,10,20\}$}
\label{f:pima:b24}
\end{figure}

Figure \ref{f:pima:b24} illustrates how mass concentrates on the MAP
in two disparate cases for varying values of $\kappa$.  For $\beta_2$
({\em left} panel), which is decidedly non-zero in the power
posterior(s), the convergence to the MAP (apparently around $\beta_2 =
6$) is modest.  In the case of $\beta_4$ ({\em right} panel) the
convergence to the MAP (to zero) is more rapid as $\kappa$ is
increased, allowing for confident variable de-selection in a way
similar to the lasso for linear regression.

Finally, we consider the case where $\nu$ is also inferred by MCMC,
jointly with the other parameters in the model.  We use the IG prior
on $\nu$ with $(r=2, d=0.1)$, a typical default choice for linear
regression \citep[e.g.,][]{gramacy:pantaleo:2010}.
\begin{figure}[ht!]
\centering
\includegraphics[scale=0.75,trim=15 0 15 40]{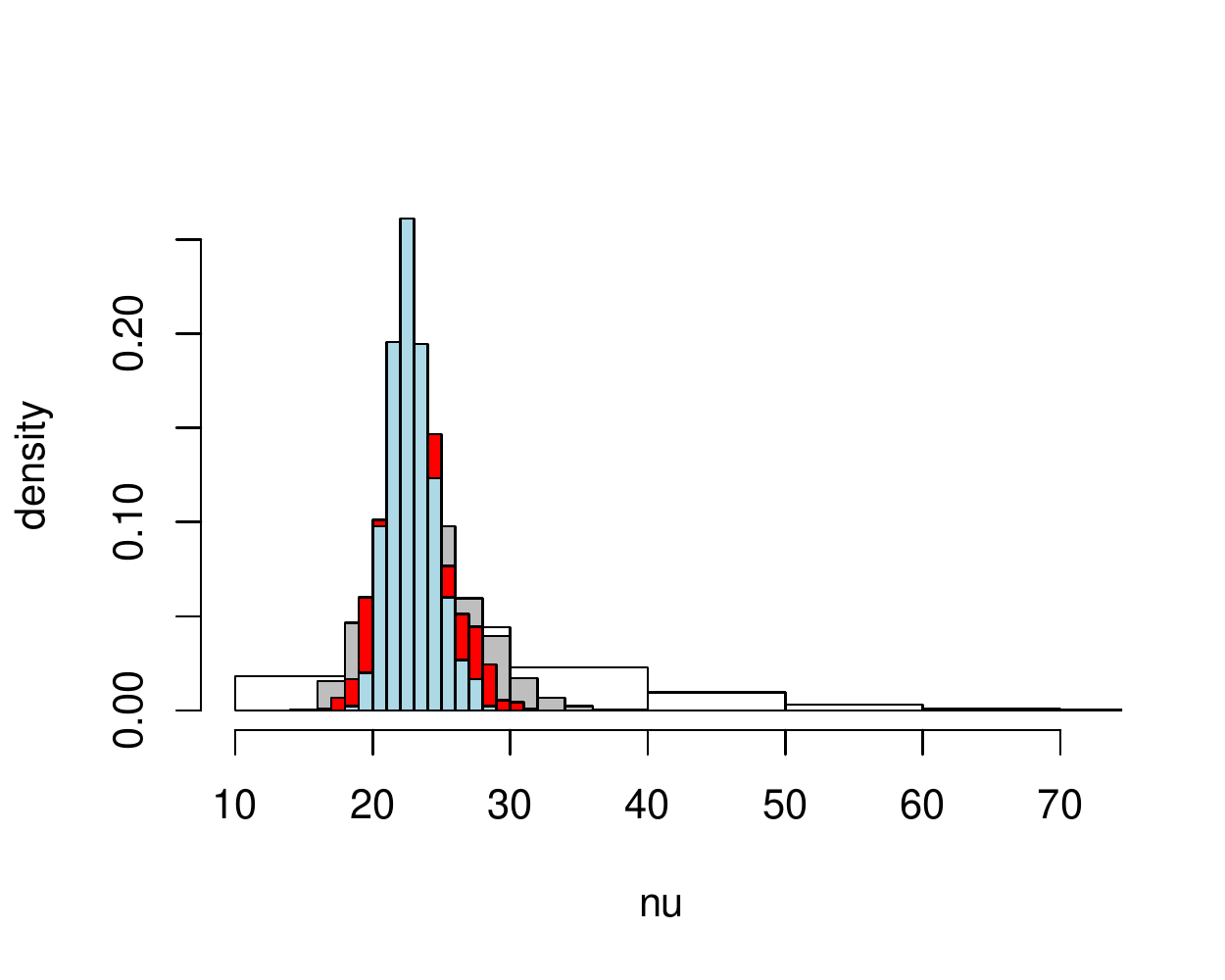}
\vspace{-0.5cm}
\caption{Concentration of posterior mass of
  $\nu$ on the Pima Indian data for $\kappa \in \{1,5,10,20\}$.  The
  histogram extends to $\nu = 100$ when $\kappa=1$, but the figure is
  trimmed.}
\label{f:pima:nu}
\end{figure}
Figure \ref{f:pima:nu} shows the marginal posterior for $\nu$ under
our settings of $\kappa$. 
The rate of convergence is modest, with the spread of samples in the
$\kappa = 20$ case being only half that of the $\kappa = 1$ case.

\subsection{Comparing c/pdf representations on binomial data}
\label{sec:st}

To illustrate the efficient handling of binomial data and,
simultaneously, to compare the cdf and pdf representations, consider
the following simple binomial logistic regression problem.  The {\em
  true} linear predictor is $\eta_i = 1 + x_i^\top \beta$ where $\beta
= (2,-3,2,-4,0,0,0,0,0)^\top$, and the $p=9$ dimensional $x_i$ are
uniform in $[0,1]^p$.  The responses, $y_i\in \{0,\dots,n_i\}$, are
sampled with $y_i|x_i \sim \mathrm{Bin}(n_i, \mu_i)$ where $n_i = 20$
and $\mu_i = e^{\eta_i}/(1+e^{\eta_i})$.

\begin{table}[ht!]
\vspace{0.5cm}
\centering
\begin{tabular}{r||rr|rr}
& \multicolumn{4}{c}{RMSE (sd)} \\
& \multicolumn{2}{c}{flat} & \multicolumn{2}{|c}{multi} \\
\hline \hline
cdf    & 0.2117 & (0.0602) & 0.2120 & (0.0606) \\   
pdf    & 0.2119 & (0.0613) & 0.2121 & (0.0602)
\end{tabular}
\hfill
\begin{tabular}{r||rr|rr}
& \multicolumn{4}{|c}{time (sd)} \\
& \multicolumn{2}{c}{flat} & \multicolumn{2}{|c}{multi} \\
\hline \hline
cdf    & 570.4 & (37.8) & 64.6 & (0.82) \\   
pdf    & 570.2 & (28.7) & 64.4 & (0.99)
\end{tabular}
\caption{Comparing RMSEs ({\em left}) and timings in seconds 
  ({\em right}) of c/pdf representations and 
  flattened/multiplicity treatments 
  of binomial regression modeling. }
\label{t:binom}
\end{table}

Table \ref{t:binom} compares four different implementations of
regularized binomial logistic regression ($\alpha = 1$) based on the
output of 100 repeated experiments with $\sum n_i = 2000$ (i.e.,
$m=100$ distinct $x_i$ predictors).  The metrics for comparison are
root mean squared error (RMSE) between the true and posterior mean
$\beta$s, and overall computing time of the respective MCMC samplers.
In all cases, we use $T=1000$ MCMC rounds with MH sampling of
$\lambda_i$ at thinning level(s) set by $\kappa'$ (i.e., via
$\kappa'_i$ for each $\lambda_i$) as described in Section
\ref{sec:conds}.  The first 100 rounds were discarded as burn-in.  The
{\em left} table shows that there is no significant difference between
the cdf and pdf representations, or between the flattened or
multiplicity handling of binomial data, in terms of
RMSE.  
%
The {\em right} table portrays a more interesting story in terms of CPU
times.  The many fewer latent variables needed by the multiplicity
implementation leads to a much (9x) faster execution compared to
flattening, with no cost in accuracy (via RMSE). In contrast, there is
no speed gain to using $n$ fewer latent $z_i$ variables in the pdf
representation.

\begin{figure}[ht!]
\centering
\includegraphics[scale=0.46,trim=0 10 0 50]{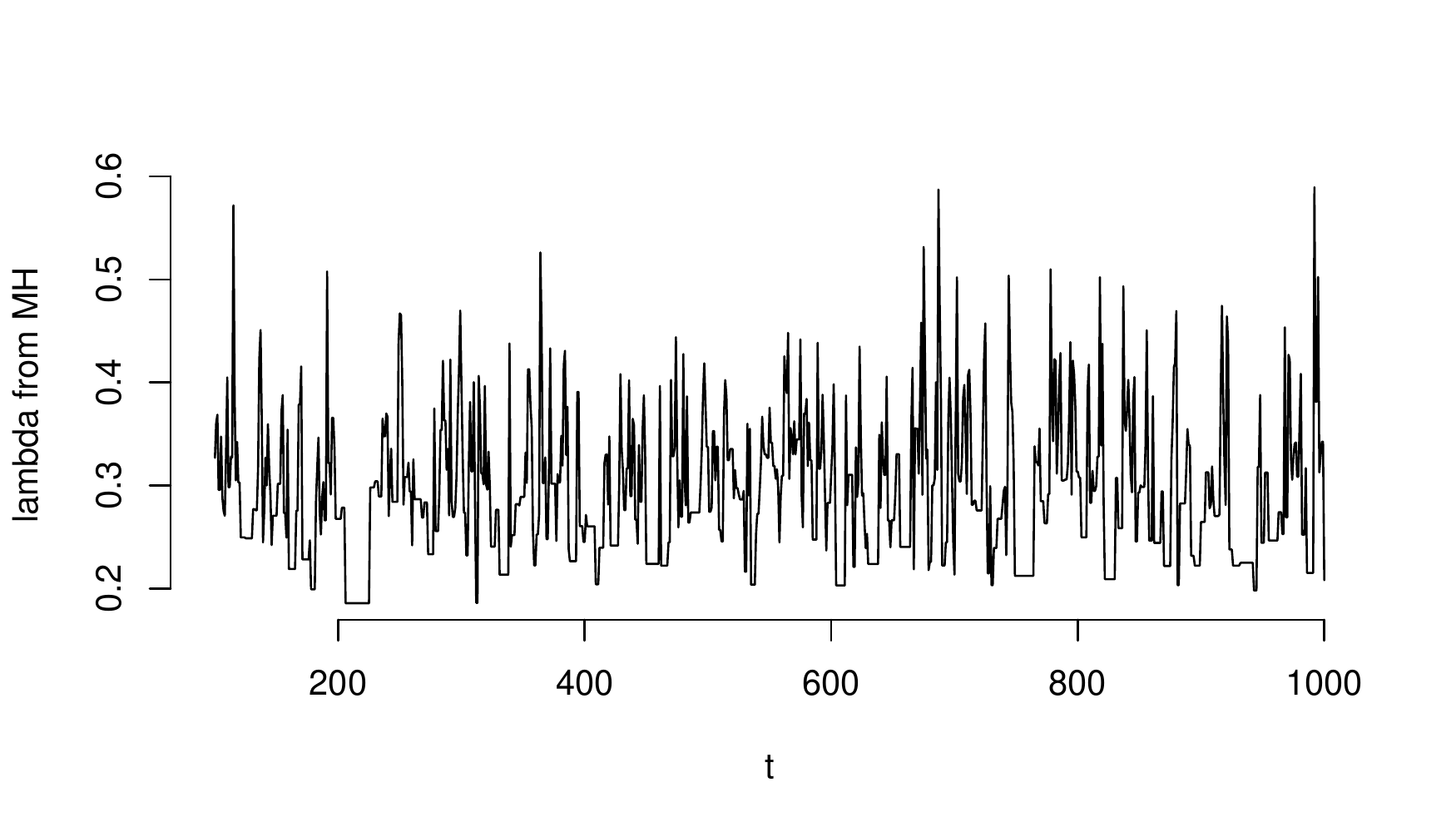}
\includegraphics[scale=0.46,trim=0 10 0 50]{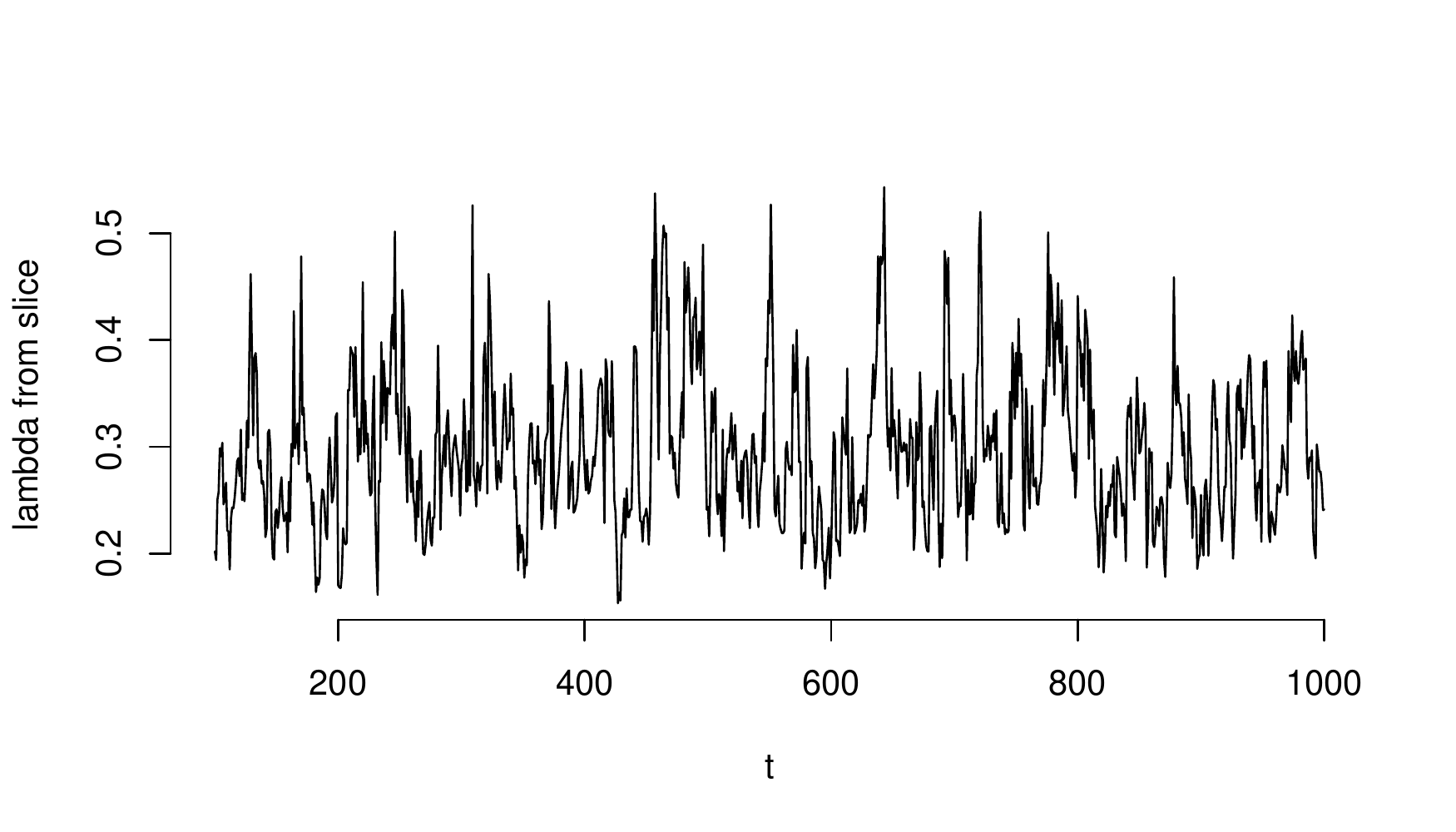}
\caption{Comparing MH ({\em left}) and slice ({\em right}) samplers
  for a $\lambda_i$ in the pdf representation.}
\label{f:slicevmh}
\end{figure}

Figure \ref{f:slicevmh} illuminates the differences in behavior
between the MH and slice sampler for the $\lambda_i$ draws (in the pdf
representation).  A particularly ``sticky'' case, as chosen from
output of the experiment, had $\kappa'_i = 14$.  The {\em top} panel
shows that many proposals from $q_{a,b}$ can be rejected under the MH
ratio, even when the chain is automatically thinned. 
The {\em bottom} panel shows the chain obtained for the same
$\lambda_i$ under the slice sampler, which never saves any rejected
draws.  However, this comes at the expense of many rejections in the
inner--loop of the slice, resulting in a slow overall sampler.  The
median was four, but the mean was 81 owing to a heavy right-hand tail
in the distribution of rejections whose central 95\% quantile spanned
to 114 and maximum reached 140,600.  The overall MCMC scheme based on
the slice sampler took four times longer than the one based on MH.
Despite the absence of rejections, the mixing in slice sampler chain
(assessed visually) was no better than MH.  Indeed, their effective
sample size due to autocorrelation \citep{kass:1998} was nearly
identical: 223 for slice sampling, and 221 for MH.  Therefore, MH is
recommended for speed considerations.

\subsection{A simulated $p \gg n$ experiment}
\label{sec:pggn}

We turn now to a predictive comparison of the methods of this paper,
both fully Bayesian and full/joint MAP (including $\nu$), benchmarked
against other modern approaches to regularized logistic regression.
Consider a synthetic data experiment like the one in Section
\ref{sec:st} except: $n_i = 5$ for each of 20 unique predictors $x_i$,
so that $\sum n_i = 100$.  Three variations on the data-generating
$\beta$ vectors were used.  In the first case $p=9$ and $\beta =
(2,-3,0.74,-0.9,0,0,0,0)^\top$; in the second case $p=100$, augmenting
$\beta$ from the first case with 91 more zeros; and in the third
$p=1000$ with 900 more zeros still.  Each experiment involves a new
random training design in the unit $p$-cube.  Random testing set are
created similarly, except that $n_i' = 100$ so $\sum n_i' = 10000$.
The metrics of comparison are (approximated) expected log likelihood
(ELL)\footnote{Specifically, the average of $(1-p_i)\log (1-\hat{p}_i)
  + p_i\log \hat{p}_i$ over all testing locations $i$, where $p_i$ and
  $\hat{p}_i$ are the true and estimated predictive probabilities of
  the first label, respectively.}  and misclassification rates.

Fully Bayesian posterior mean estimators (i.e., $\kappa = 1$) are
derived via priors/MCMC exactly as described in the preceding sections
with $(100,1000)$, $(500,1500)$, $(1000,2000)$ burn-in and total MCMC
rounds in each of the cases $p=9,100,1000$, respectively.  MAP
estimators are found by running a $\kappa = 10$ chain initialized at
$(\beta, \lambda, \nu)$-values from the $\kappa = 1$ chain used for
the mean estimators, except in the $p=1000$ case where $\nu$ was fixed
to its posterior mean for reasons laid out in Section
\ref{sec:anneal}.  Comparators include: the MLE obtained via the {\tt
  glm} command in {\sf R}; a binomial fit from the {\tt glmnet}
package \citep{fried:hast:tibsh:2009}; and the estimator of
\cite{krish:etal:2005}\footnote{This is equivalent to the
  \citet{genkin:lewis:madigan:2007} estimator but computationally less
  efficient.} [``krish'' for short].  The MLE was unstable in the
$p=100$ \& $1000$ cases, so these results were omitted.  CV was used
to choose the penalty parameter in the $p=9$ \& $100$ cases for {\tt
  glmnet}, via {\tt cv.glmnet}.  The same procedure gave fatal errors
in the $p=1000$ case so we plugged in the estimate obtained from the
corresponding $p=100$ run in for this final case.  Reliably setting
the penalty parameter for ``krish'', via CV or otherwise, was too
computationally intensive for the $p=100,1000$ cases so we picked a
setting by hand using out-of-sample simulations from the $p=9$ case.

\begin{figure}[ht!]
\centering
\begin{minipage}{11cm}
\includegraphics[scale=0.75,trim=20 0 0 0]{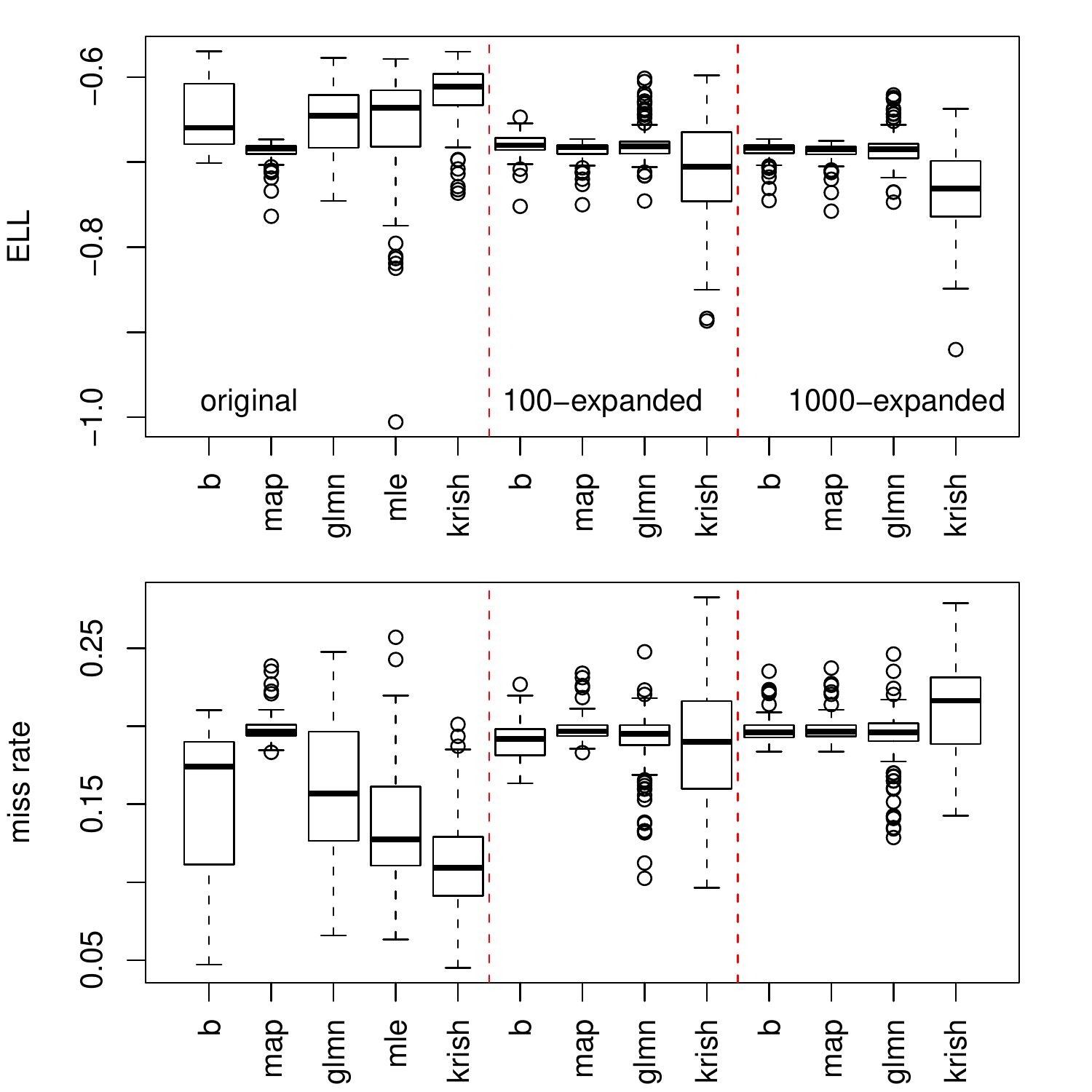}
\end{minipage}
\begin{minipage}{5cm}
\footnotesize
\begin{tabular}{l||rrr}
ELL & 5\% & avg & 95\%\\
\hline \hline
b       & -0.694  &  -0.646  & -0.582 \\
map   & -0.710 &   -0.688 & -0.677 \\
glmn  & -0.703 &  -0.650  & -0.593 \\
mle    & -0.797 & -0.658   & -0.588 \\
krish  & -0.699 & -0.619   & -0.579 \\
\hline 
b       & -0.701 & -0.680   & -0.660 \\
map  & -0.712 &  -0.687  & -0.678 \\
glmn & -0.705 &  -0.678  & -0.629 \\
krish & -0.815 & -0.711   & -0.628 \\
\hline
b      & -0.707 & -0.687    & -0.676 \\
map  & -0.711 &  -0.689  & -0.677 \\
glmn & -0.707 & -0.683   & -0.639 \\
krish & -0.707 & -0.734   & -0.651
\end{tabular}

\vspace{0.5cm}

\begin{tabular}{l||rrr}
miss   & 5\% & avg & 95\% \\
\hline \hline
b        & 0.065 & 0.152   & 0.201 \\
map   & 0.189 &  0.199  & 0.212 \\ 
glmn  & 0.092 & 0.159   & 0.217 \\
mle    & 0.074 & 0.136   & 0.213 \\
krish   & 0.061 & 0.113  & 0.184 \\
\hline
b        & 0.172 & 0.191   & 0.210 \\
map   & 0.188 &  0.199  & 0.212 \\ 
glmn  & 0.133 & 0.189   & 0.212 \\
krish  & 0.129 & 0.189   & 0.244 \\
\hline
b       & 0.187 & 0.198    & 0.215 \\
map  & 0.188 & 0.198    & 0.215 \\
glmn & 0.142 & 0.193    & 0.214 \\
krish & 0.151 & 0.211    & 0.252
\end{tabular}
\end{minipage}
\caption{Expected log likelihood (ELL) and misclassification rates in
boxplot ({\em left}) and tabular ({\em right}) form.  In both cases
there are three sections, depending on the number of irrelevant
predictors in the design matrix, wherein the same estimators are
applied.  The vertical  dashed-red lines in the boxplots indicate
the same demarkation as the horizontal lines in the tables.}
\label{f:stellmiss}
\end{figure}

The results of the Monte Carlo experiment are summarized in Figure
\ref{f:stellmiss} by boxplots, and numerically.  The best estimators
have high ELL, low miss rates, and lower variability across the 100
repetitions. The fully Bayesian and ``krish'' methods are the best
when $p=9$ ({\em left}-hand region of the boxplots and the {\em top}
region of the tables).  The former wins by ELL, having fewer low
values, and the latter wins on miss rate, having more small ones.  The
``krish'' method wins by both metrics on average, since it employs a
fortuitously hand-chosen setting of the penalty parameter.  The MLE is
good on average, but has some extreme ELL and miss rate values.  The
{\tt glmnet} and MAP estimators are positioned in
between. 

Distinctions in performance between the methods increases with $p$.
See the {\em right}-hand regions of the boxplots and the {\em bottom}
regions of the tables.  The ``krish'' method suffers from high
variability due to the fixed choice of the penalty parameter.  The
{\tt glmnet} variability is much lower, but there are many extreme
outliers.  Behavior in both $p=100$ and $1000$ cases is qualitatively
similar for this estimator even though the former used CV to set the
penalty parameter and the latter used the same fixed value.  The MAP
and fully Bayesian estimators have similar average behavior compared
to other estimators, but with lower variability.  Apparently, choosing
the penalty parameter via the posterior offers the most stability in
high dimensional settings.  The fully Bayesian approach appears
preferable to the MAP in all cases, but this distinction is harder to
make out as $p$ increases.

\subsection{Spam data with interactions}
\label{sec:spam}

For a similar real-data experiment, consider the Spambase data set
from UCI.  It contains the binary classifications of 4601 emails based
on 57 attributes which are treated as predictors.  An
interaction-expanded version of the predictor set contains
approximately 1700 predictors.  We performed a Monte Carlo experiment
comprising of 20 random 5-fold CV training and testing sets using both
the original and expanded predictors.  Estimators were fit on the 100
training sets, and validated by misclassification rate on the testing
ones.  The Bayes estimators used (500,1500) MCMC (burn-in, total)
rounds with the original 57 attributes, and (1000,2000) with the
expanded set.  The MAP and {\tt glmnet} calculations were exactly as
described for the $p=100$ case in Section \ref{sec:pggn} for the
original predictor set, and like the $p=1000$ case for the expanded
one.  And ``krish'' was like $p=9$ and $p=100$, respectively.  

\begin{figure}[ht!]
\centering
\begin{minipage}{9cm}
\includegraphics[scale=0.6,trim=20 0 0 0]{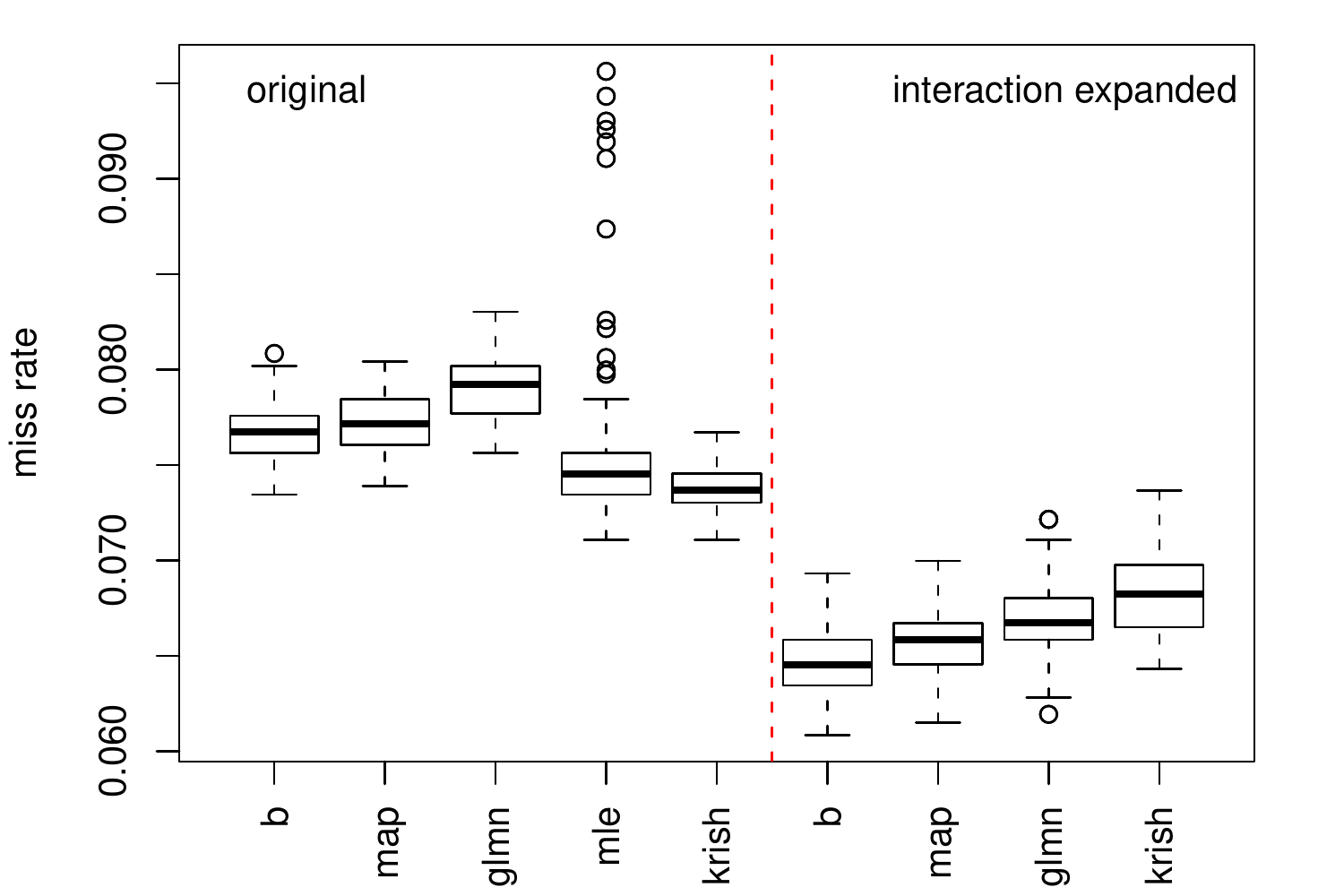}
\end{minipage}
\begin{minipage}{5cm}
\footnotesize
\begin{tabular}{l||rrr}
miss   & 5\% & avg & 95\% \\
\hline \hline
b       & 0.074 &  0.077 & 0.079 \\
map  & 0.074 &  0.077 & 0.080 \\
glmn & 0.077 &  0.079 & 0.082 \\
mle   & 0.071 &  0.076 & 0.092 \\
krish & 0.072 &  0.074 & 0.076 \\
\hline
b       & 0.062  & 0.065 & 0.068 \\
map  & 0.063  & 0.066 & 0.068 \\
glmn & 0.064  & 0.067 & 0.070 \\
krish & 0.065  & 0.068 & 0.072 
\end{tabular}
\end{minipage}
\caption{Misclassification rates in boxplot ({\em left}) and tabular
  ({\em right}) form.  In both cases there are two sections, depending
  absence or presence of interaction terms in the design matrix,
  wherein the same estimators are applied.  The vertical dashed-red
  line in the boxplot indicates the same demarkation as the
  horizontal line in the table.}
\label{f:spammiss}
\end{figure}

The results of the experiment are summarized in Figure
\ref{f:spammiss}.  The first thing we notice is that, in contrast with
the results in Section \ref{sec:pggn}, the performance improves as the
predictor set expands since some of the interaction terms make good
predictors.  
The MLE is unstable, and so the regularized estimators
offer an improvement even when the number of predictors is small
relative to the number of instances.  The Bayesian methods
unilaterally outperform {\tt glmnet}, and using the posterior to set
the value of the regularization parameter is important in high
dimensional settings.  The ``krish'' estimator with fortuitous
regularization is the best on the original predictor set, but worst on
the expanded one where a revised setting of regularization 
could not be automated efficiently.

\section{Discussion and extension}
\label{sec:discuss}

We provide a simulation-based approach to regularized logistic
regression that facilitates a variety of inferential goals under a
single framework.  Most of the development of the methodology, and all
of the applications, involved the $\alpha = 1$ case.  Everything
extends to the ridge prior ($\alpha = 2$), i.e., an independent normal
prior for each coefficient $\beta_j$ with variance $\sigma_j^2 \nu^2 /
\kappa$.  Then, $p_\kappa(\omega_j |\beta, \nu)$ is a point mass at
$\omega_j = 1$.  Thus similar conjugacy results hold for the gamma
prior on $\nu$ and
$\nu^2$. 

From a computational perspective, our methods are competitive with the
state-of-the art in un-regularized (and $\kappa = 1$) contexts too.
For example, we compared the efficiency of our methods to the ``dRUM''
MH sampler described by \cite{frueh2:2010}.  This method is attractive
because it is fast and easy to implement.  For example, on the Pima
data it takes about 32s to generate 10,000 samples from the posterior
which is about 7x faster than our pdf representation, which took 230s.
However, the MH acceptance rate of the dRUM method was 46\% which lead
to an marginal ESS of 957 averaged over the nine $\beta_j$
coefficients.  Our pdf representation had an average ESS that was
about 5x better, at 4518.  So the methods work out to have similar
overall efficiences in that example.  But in higher dimension like the
57-d spam data, our Gibbs sampling approach is much more attractive.
The acceptance rate for dRUM was extremely low at 0.4\%, which leads
to ESSs that are essentially nil.  Although our pdf representation is
(again) 7x slower, faster convergence due to better movement in the
chain leads to reasonable ESSs around 500.

There are several extensions of our methodology that readily
present themselves.  For example, handling polychotomous data (i.e.,
$>2$ classes) is straightforward.  Following the setup in HH we may
introduce $C$ collections of coefficients $\beta^{(1)}, \dots,
\beta^{(C)}$ for $C$ classes with the convention that $\beta^{(C)} =
0$ so that logistic regression is recovered in the $C=2$ case.  Then,
we simply work with the conditional likelihoods $L(\beta^{(j)} | y,
\beta^{(-j)})$ which turn out to have exactly the form of a logistic
regression likelihood for the class indicator that each $y_i = j$,
independently for $i=1,\dots, n$.  If there are $n_i > 1$ {\em trials}
for predictors $x_i$, then our algorithm for binomial logistic
regression is applicable via a vectorized multiplicity parameter as
described in Section
\ref{sec:binom}.  
Extending the methods to ordinal responses is even
easier. \citet[][Chapter 4]{johnson:albert:1999} describe a Bayesian
probit model which may be adapted for the logit case following either
HH or our cdf representation.  The pdf representation may not be
readily applicable because the latent $z_i$ are useful for efficient
sampling of the ordinal break
points.  

An further direction is to other classes of regularization priors.
Implementing the Normal--Gamma extension \citep{griffin:brown:2010}
requires adding an extra (conjugate) parameter.  A promising new
approach is the {\em horseshoe} prior
\citep{carvalho:polson:scott:2010}, which can be implemented with the
addition of a slice sampler.  Often variable selection is a primary
goal of regularization, for which our methods would require further
extension.  For example, HH describe an approach to variable selection
for logistic regression via Reversible Jump MCMC \citep{green:1995}
which is adaptable to our framework.
A similar regularized approach in a linear regression is provided
by\cite{gramacy:pantaleo:2010}.  For variable selection for logistic
regression using spike-and-slab priors, see \citet{tuchler:2008}.


\subsection*{Acknowledgments}

This research was partially funded by EPSRC grant EP/D065704/1 to RBG.
The authors would like thank Matt Taddy for interesting discussions on
the efficient handling of Binomial data, extensions to Multinomial
regression, and EM code for the MAP estimator(s).  We are grateful to
two referees and an associate editor for valuable comments.

\appendix 

\section{Posterior conditional for $\beta$ in the pdf representation}
\label{sec:beta}

For a particular $\lambda$, i.e., ignoring the integral in Eq.~(\ref{eq:pdf}),
we have the following expression for the likelihood in vector/matrix form.
\[
\prod_{i=1}^n \left(1 + e^{- y_i x_i^\top\beta}\right) = e^{a y^\top X
  \beta}
\exp \left\{ - \frac{1}{2}\left((y.X)\beta +
    \frac{1}{2}(a-b)\lambda\right)^\top \!\Lambda^{-1} 
\left((y.X)\beta + \frac{1}{2}(a-b)\lambda\right) \right\}
\]
An expression for the posterior conditional for $\beta$ can then
obtained by multiplying by the kernel of the MVN prior given $\omega$,
provided below Eq.~(\ref{eq:alpha-cdlike}), namely: $\exp\{-
\frac{1}{2} \beta^\top \!(\frac{\kappa^2}{\nu^2} \Sigma^{-1} \Omega^{-1}
\beta)\}$.  Combining the terms in the three exponents gives the
following quadratic form.
\[
-\frac{1}{2}\left[ - 2 a y^\top X \beta + \!\left( (y.X)\beta +
    \frac{1}{2}(a-b)\lambda \right)^{\!\top} \!\!\!\Lambda^{-1} \!\left( (y.X)\beta +
    \frac{1}{2}(a-b)\lambda \right)\! +
  \beta^\top\!\left(\frac{\kappa^2}{\mu^2} \Sigma^{-1} \Omega^{-1}\!\right)\! \beta
\right]
\]
Collecting terms for $\beta$ yields
\[
\beta^\top\left((y.X)^\top \Lambda^{-1} (y.X) + \frac{\kappa^2}{\nu^2}
\Sigma^{-1} \Omega^{-1}\right) \beta - 
(2 a y^\top X - (a-b)(y.X) \Lambda^{-1} \lambda) \beta.
\]
Therefore we deduce that the conditional is
$\mathcal{N}_p(\tilde{\beta}, V)$ where $V^{-1} = (y.X)^\top
\Lambda^{-1} (y.X) + \frac{\kappa^2}{\nu^2} \Sigma^{-1} \Omega^{-1}$.
Recognizing that $(y.X) \Lambda^{-1} \lambda = X^\top y$ gives that
$\tilde{\beta} = V(a - \frac{1}{2}[a-b]) X^\top y$.

\section{Generalized Inverse Gaussian distribution}
\label{sec:gig}

The pdf of a Generalized Inverse Gaussian,
$\gig(\lambda, \chi, \psi)$ is
\[
g(x; \lambda, \chi, \psi) =
\frac{(\psi/\chi)^{\lambda/2}}{2K_\lambda(\sqrt{\psi \chi})}
    x^{\lambda-1} \exp\left\{-\frac{1}{2}(\psi x + \chi/x) \right\},
\]
where $K_\lambda$ is a modified Bessel function of the second kind.
If $ X \sim \gig\!\left( \frac{1}{2} , \chi , \psi \right) $ then $
X^{-1} \sim \ig(\mu=\sqrt{\psi/\chi}, \lambda=\psi)$ where 
where $\ig$ is the inverse Gaussian distribution with
pdf
\[
f( x; \mu,\lambda)  = \sqrt{ \frac{\lambda}{2\pi x^3} }
\exp\left\{ -\frac{\lambda \left(  x-\mu\right)^{2} }{2 \mu^2 x } 
\right\}.
\]
The mean and variance are $\mathbb{E}\{ x\} =\mu$ and
$\mathrm{\mathbb{V}ar}[ x] =\mu^{3}/ \lambda$.  A generalized inverse
Gaussian $\gig\left(\frac{1}{2} , \chi , \psi \right)$ is an inverse
of an Inverse Gaussian.  For simulation from $\gig$ and $\ig$
distributions see \cite{devroye:1986}.

\bibliography{reglogit}

\begin{thebibliography}{37}
\newcommand{\enquote}[1]{``#1''}
\expandafter\ifx\csname natexlab\endcsname\relax\def\natexlab#1{#1}\fi

\bibitem[\protect\citename{Andrews and Mallows, }1974]{andrews:mallows:1974}
Andrews, D. and Mallows, C. (1974).
\newblock \enquote{Scale Mixtures of Normal Distributions.}
\newblock {\em Journal of the Royal Statistical Soceity, Series B\/}, 36,
  99--102.

\bibitem[\protect\citename{Asuncion and Newman, }2007]{Asuncion+Newman:2007}
Asuncion, A. and Newman, D. (2007).
\newblock \enquote{{UCI} Machine Learning Repository.}

\bibitem[\protect\citename{Barndorff-Nielsen et~al.,
  }1982]{bn:kent:sorensen:1982}
Barndorff-Nielsen, O., Kent, J., and Sorensen, M. (1982).
\newblock \enquote{Normal Variance-Mean Mixtures and $z$-distributions.}
\newblock {\em International Statistical Review\/}, 50, 145--159.

\bibitem[\protect\citename{Bernstein, }2005]{berns:2005}
Bernstein, D. (2005).
\newblock {\em Matrix Mathematics\/}.
\newblock Princeton, NJ: Princeton University Press.

\bibitem[\protect\citename{Box and Tiao, }1973]{box:tiao:1973}
Box, G. and Tiao, G. (1973).
\newblock {\em Bayesian Inference in Statistical Analysis\/}.
\newblock Mass: Addison Wesley.

\bibitem[\protect\citename{Carlin and Polson, }1991]{carlin:polson:1991}
Carlin, B.~P. and Polson, N.~G. (1991).
\newblock \enquote{Inference for Nonconjugate Bayesian Models using the Gibbs
  sampler.}
\newblock {\em The Canadian Journal of Statistics\/}, 19, 4, 399--405.

\bibitem[\protect\citename{Carvalho et~al., }2010]{carvalho:polson:scott:2010}
Carvalho, C., Polson, N., and Scott, J. (2010).
\newblock \enquote{The horseshoe estimator for sparse signals.}
\newblock {\em Biometrika\/}, 9, 2, 465--480.

\bibitem[\protect\citename{Devroye, }1986]{devroye:1986}
Devroye, L. (1986).
\newblock {\em Non-Uniform Random Variate Generation\/}.
\newblock Springer-Verlag.

\bibitem[\protect\citename{Doucet et~al., }2002]{doucet:godsill:robert:2002}
Doucet, A., Godsill, S., and Robert, C. (2002).
\newblock \enquote{Marginal maximum a posteriori estimation using Markov chain
  Monte Carlo.}
\newblock {\em Statistics and Computing\/}, 21, 77--84.

\bibitem[\protect\citename{Fahrmeir et~al., }2010]{fahrmeir:etal:2010}
Fahrmeir, L., Kneib, T., and Konrath, S. (2010).
\newblock \enquote{Bayesian regularisation in structured additive regression: A
  unifying perspective on shrinkage, smoothing and predictor selection.}
\newblock {\em Statistics and Computing\/},  203--219.

\bibitem[\protect\citename{Friedman et~al., }2010]{fried:hast:tibsh:2009}
Friedman, J.~H., Hastie, T., and Tibshirani, R. (2010).
\newblock \enquote{Regularization Paths for Generalized Linear Models via
  Coordinate Descent.}
\newblock {\em Journal of Statistical Software\/}, 33, 1, 1--22.

\bibitem[\protect\citename{Friel and Pettitt, }2008]{friel:pettitt:2008}
Friel, N. and Pettitt, A. (2008).
\newblock \enquote{Marginal likelihood estimation via power posteriors.}
\newblock {\em Journal of the Royal Statistical Society, Series B.\/}, 70, 3,
  589--607.

\bibitem[\protect\citename{Fr\"uhwirth-Schnatter and Fr\"uhwirth,
  }2007]{fruh2:2007}
Fr\"uhwirth-Schnatter, S. and Fr\"uhwirth, R. (2007).
\newblock \enquote{Auxilliary Mixture Sampling with Applications to Logistic
  Models.}
\newblock {\em Computational Statistics and Data Analysis\/}, 51, 7,
  3509--3528.

\bibitem[\protect\citename{Fr\"uhwirth-Schnatter and Fr\"uhwirth,
  }2010]{frueh2:2010}
--- (2010).
\newblock \enquote{Data augmentation and MCMC for binary and multinomial logit
  models.}
\newblock In {\em Statistical Modelling and Regression Structures –
  Festschrift in Honour of Ludwig Fahrmeir\/}, eds. T.~Kneib and G.~Tutz,
  111--132. Physica-Verlag.

\bibitem[\protect\citename{Fr\"uhwirth-Schnatter et~al.,
  }2009]{frueh2:etal:2009}
Fr\"uhwirth-Schnatter, S., R., Fr\"uhwirth, Held, L., and Rue, H. (2009).
\newblock \enquote{Improved auxiliary mixture sampling for hierarchical models
  of non-Gaussian data.}
\newblock {\em Statistics and Computing\/}, 19, 479--492.

\bibitem[\protect\citename{Genkin et~al., }2007]{genkin:lewis:madigan:2007}
Genkin, A., Lewis, D., and Madigan, D. (2007).
\newblock \enquote{Large-Scale {B}ayesian Logistic Regression for Text
  Categorization.}
\newblock {\em Technometrics\/}, 49, 3, 291--304.

\bibitem[\protect\citename{Godsill, }2000]{godsill:2000}
Godsill, S. (2000).
\newblock \enquote{Inference in symmetric alpha-stable noise using {MCMC} and
  the slice sampler.}
\newblock In {\em IEEE International Conference on Acoustics, Speech and Signal
  Processing\/}, vol.~VI,  3806--3809.

\bibitem[\protect\citename{Gramacy and Pantaleo, }2010]{gramacy:pantaleo:2010}
Gramacy, R. and Pantaleo, E. (2010).
\newblock \enquote{Shrinkage regression for multivariate inference with missing
  data, and an application to portfolio balancing.}
\newblock {\em Bayesian Analysis\/}, 5, 2, 237--262.

\bibitem[\protect\citename{Green, }1995]{green:1995}
Green, P. (1995).
\newblock \enquote{Reversible Jump Markov Chain Monte Carlo Computation and
  Bayesian Model Determination.}
\newblock {\em Biometrika\/}, 82, 711--732.

\bibitem[\protect\citename{Griffin and Brown, }2010]{griffin:brown:2010}
Griffin, J.~E. and Brown, P.~J. (2010).
\newblock \enquote{Inference with Normal--Gamma prior distributions in
  regression problems.}
\newblock {\em Bayesian Analysis\/}, 5, 1, 171--188.

\bibitem[\protect\citename{Hans, }2009]{hans:2009}
Hans, C. (2009).
\newblock \enquote{Bayesian Lasso Regression.}
\newblock {\em Biometrika\/}, 96, 836--845.

\bibitem[\protect\citename{Holmes and Held, }2006]{holmes:held:2006}
Holmes, C. and Held, K. (2006).
\newblock \enquote{Bayesian Auxilliary Variable Models for Binary and
  Multinomial Regression.}
\newblock {\em Bayesian Analysis\/}, 1, 1, 145--168.

\bibitem[\protect\citename{Jacquier et~al.,
  }2007]{jacquier:johannes:polson:2007}
Jacquier, E., Johannes, M., and Polson, N. (2007).
\newblock \enquote{MCMC Maximum Likelihood for Latent State Models.}
\newblock {\em Journal of Econometrics\/}, 137, 615--640.

\bibitem[\protect\citename{Johnson and Albert, }1999]{johnson:albert:1999}
Johnson, V. and Albert, J. (1999).
\newblock {\em Ordinal Data Modeling\/}.
\newblock Springer-Verlag.

\bibitem[\protect\citename{Kass et~al., }1998]{kass:1998}
Kass, R.~E., Carlin, B.~P., Gelman, A., and Neal, R.~M. (1998).
\newblock \enquote{Markov Chain Monte Carlo in Practice: A Roundtable
  Discussion.}
\newblock {\em The American Statistician\/}, 52, 2, 93--100.

\bibitem[\protect\citename{Kirkpatrick et~al., }1983]{kirkp:etal:1984}
Kirkpatrick, S., Gelatt, C., and Vecci, M. (1983).
\newblock \enquote{Optimization by simulated annealing.}
\newblock {\em Science\/}, 220, 671--680.

\bibitem[\protect\citename{Krishnapuram et~al., }2005]{krish:etal:2005}
Krishnapuram, B., Carin, L., Figueiredo, M., and Hartemink, A. (2005).
\newblock \enquote{Sparse Multinomial Logistic Regression: Fast Algorithms and
  Generalization Bounds.}
\newblock {\em IEEE Pattern Analysis and Machine Intellegence\/}, 27, 6,
  957--969.

\bibitem[\protect\citename{Madigan and Ridgeway, }2004]{madigan:ridgeway:2004}
Madigan, D. and Ridgeway, G. (2004).
\newblock \enquote{Discussion of `Least Angle Regression' by B. Efron, T.
  Hastie, I. Johnstone, and R. Tibshirani.}
\newblock {\em Annals of Statistics\/}, 32, 2, 465--469.

\bibitem[\protect\citename{Park and Hastie, }2008]{park:hastie:2008}
Park, M. and Hastie, T. (2008).
\newblock \enquote{Penalized Logistic Regression for Detecting Gene
  Interactions.}
\newblock {\em Biostatistics\/}, 9, 1, 30--50.

\bibitem[\protect\citename{Park and Casella, }2008]{park:casella:2008}
Park, T. and Casella, G. (2008).
\newblock \enquote{The Bayesian Lasso.}
\newblock {\em Journal of the American Statistical Association\/}, 103, 482,
  681--686.

\bibitem[\protect\citename{Pincus, }1968]{pincus:1968}
Pincus, M. (1968).
\newblock \enquote{A Closed Form Solution of Certain Programming Problems.}
\newblock {\em Operations Research\/}, 18, 1225--1228.

\bibitem[\protect\citename{Robert, }1995]{robert:1995}
Robert, C. (1995).
\newblock \enquote{Simulation of Truncated Normal Variables.}
\newblock {\em Statistics and Computing\/}, 5, 2, 121--125.

\bibitem[\protect\citename{{{\sf R} Development Core Team}, }2009]{cran:R}
{{\sf R} Development Core Team} (2009).
\newblock {\em {\sf R}: A Language and Environment for Statistical
  Computing\/}.
\newblock {\sf R} Foundation for Statistical Computing, Vienna, Austria.
\newblock {ISBN} 3-900051-07-0.

\bibitem[\protect\citename{Tibshirani, }1996]{tibsh:1996}
Tibshirani, R. (1996).
\newblock \enquote{Regression shrinkage and Selection via the Lasso.}
\newblock {\em Journal of the Royal Statistical Society, Series B.\/}, 58, 1,
  267--288.

\bibitem[\protect\citename{T\"uchler, }2008]{tuchler:2008}
T\"uchler, R. (2008).
\newblock \enquote{Bayesian variable selection for logistic models using
  auxiliary mixture sampling.}
\newblock {\em Journal of Computational and Graphical Statistics\/}, 17,
  76--94.

\bibitem[\protect\citename{Weron, }1996]{weron:1996}
Weron, R. (1996).
\newblock \enquote{On the {C}hambers-{M}allows-{S}tuck Method for Simulating
  Skewed Stable Random Variables.}
\newblock {\em Statistics and Probability Letters\/}, 28, 2, 165--171.

\bibitem[\protect\citename{West, }1987]{west:1987}
West, M. (1987).
\newblock \enquote{On Scale Mixtures of Normal Distributions.}
\newblock {\em Biometrika\/}, 74, 3, 646--648.

\end{thebibliography}
\bibliographystyle{jasa}

\end{document}